\documentclass{article}[11pt,letter]

\usepackage[T1]{fontenc}
\usepackage[english]{babel}
\usepackage[cmex10]{amsmath}
\usepackage[utf8]{inputenc}
\usepackage{amssymb}
\usepackage{amsthm}
\usepackage{authblk}
\usepackage{mathtools}
\usepackage{xspace}
\usepackage[symbol*]{footmisc}
\usepackage[hyperfootnotes=false,breaklinks,bookmarks=false]{hyperref}
\usepackage{algorithmic}
\usepackage{algorithm}
\usepackage{enumerate}
\usepackage{color}
\usepackage{times}
\usepackage{fullpage}

\renewcommand{\algorithmiccomment}[1]{\bgroup\hfill//~#1\egroup}

\newtheorem{definition}{Definition}[section]
\newtheorem{lemma}[definition]{Lemma}
\newtheorem{theorem}[definition]{Theorem}
\newtheorem{proposition}[definition]{Proposition}

\newtheorem{observation}[definition]{Observation}
\newtheorem{corollary}[definition]{Corollary}
\newcommand{\bigo}{\mathcal{O}}
\newcommand{\state}{\mathcal{S}}
\newcommand{\inedg}{\mathrm{in}}
\newcommand{\outedg}{\mathrm{out}}
\newcommand{\ie}{{\it i.e.}\xspace}
\newcommand{\eg}{{\it e.g.}\xspace}
\newcommand{\load}{\mathcal{L}}
\newcommand{\wload}{\mathcal{W}}
\newcommand{\cload}{\mathcal{C}}
\newcommand{\dt}{\Delta t}
\newcommand{\expl}{\mathrm{expl}}
\newcommand{\walk}{\mathcal{E}}
\newcommand{\mset}[1]{\{\!\{#1\}\!\}}
\newcommand{\const}{\mathrm{const}}
\newcommand{\closedrange}[2]{[#1\,..\,#2]}
\newcommand{\halfrange}[2]{[#1\,..\,#2)}
\newcommand\defeq{\stackrel{\mathclap{\normalfont{\mbox{\text{\tiny{def}}}}}}{=}}
\newcommand{\stab}{t_{s}}
\newcommand{\per}{t_{p}}
\newcommand{\etal}{{\it et~al.}}
\newcommand{\tokens}{\mathit{tokens}}
\newcommand{\pointer}{\mathit{pointer}}

\DefineFNsymbolsTM{myfnsymbols}{
  \textasteriskcentered *
  \textdagger    \dagger
  \textdaggerdbl \ddagger
  \textsection   \mathsection
  \textbardbl    \|%
  \textparagraph \mathparagraph
}%
\setfnsymbol{myfnsymbols}

\DeclareMathOperator*{\avg}{\mathbf{avg}}

\usepackage{setspace}   

\begin{document}
\author[1]{Jérémie Chalopin}
\author[1]{Shantanu Das}
\author[2]{Paweł Gawrychowski}
\author[3]{Adrian Kosowski}
\author[1]{Arnaud Labourel}
\author[4]{Przemys\l{}aw Uznański\thanks{Part of the work was done while the author was affiliated with LIF,  CNRS and Aix-Marseille University (supported by the Labex Archimède and by the ANR project MACARON (ANR-13-JS02-0002)).}}

\affil[1]{LIF, CNRS and Aix-Marseille University, France}
\affil[2]{Max-Planck-Institut f\"{u}r Informatik, Saarbr\"ucken, Germany}
\affil[3]{Inria – LIAFA – Paris Diderot University, France}
\affil[4]{Helsinki Institute for Information Technology HIIT, Department of Computer Science, Aalto University, Finland}

\date{}

\title{Lock-in Problem for Parallel Rotor-router Walks}
\maketitle

\begin{abstract}
The \emph{rotor-router} model, also called the \emph{Propp machine}, was introduced as a deterministic alternative to the random walk. In this model, a group of identical tokens are initially placed at nodes of the graph. Each node maintains a cyclic ordering of the outgoing arcs, and during consecutive turns the tokens are propagated along arcs chosen according to this ordering in round-robin fashion.
The behavior of the model is fully deterministic. Yanovski \etal (2003) proved that a single rotor-router walk on any graph with $m$ edges and diameter $D$ stabilizes to a traversal of an Eulerian circuit on the set of all $2m$ directed arcs on the edge set of the graph, and that such periodic behaviour of the system is achieved after an initial transient phase of at most $2mD$ steps.

The case of multiple parallel rotor-routers was studied experimentally, leading Yanovski~\etal\ to the conjecture that a system of $k>1$ parallel walks also stabilizes with a period of length at most $2m$ steps. In this work we disprove this conjecture, showing that the period of parallel rotor-router walks can in fact, be superpolynomial in the size of graph. On the positive side, we provide a characterization of the periodic behavior of parallel router walks, in terms of a structural property of stable states called a \emph{subcycle decomposition}.
This property provides us the tools to efficiently detect whether a given system configuration corresponds to the transient or to the limit behavior of the system. Moreover, we provide polynomial upper bounds of $\bigo(m^4D^2 + mD\log k)$ and $\bigo(m^5k^2)$ on the number of steps it takes for the system to stabilize. Thus, we are able to predict any future behavior of the system using an algorithm that takes polynomial time and space. In addition, we show that there exists a separation between the stabilization time of the single-walk and multiple-walk rotor-router systems, and that for some graphs the latter can be asymptotically larger even for the case of $k=2$ walks.
\end{abstract}

\section{Introduction}

Dynamical processes occurring in nature provide inspiration for simple, yet powerful distributed algorithms. For example, the \emph{heat equation}, which describes real-world processes such as heat and particle diffusion, also proves useful when designing schemes for load-balancing and token rearrangement in a discrete graph scenario. In the diffusive model of load-balancing on a network, each node of the network is initially endowed with a certain load value, and in each step it distributes a fixed proportion of its load evenly among its neighbors. Given that such a balancing operation is performed for load which is infinitely divisible (so-called \emph{continuous diffusion}), in the long term the distribution of load converges on a degree-regular network to uniform over all nodes. When load is composed of indivisible unit tokens, the continuous diffusion process is no longer practicable. It is, however, possible to design randomized schemes in which the \emph{expected} value of load of each node at each moment of time corresponds precisely to the value of its load in the corresponding continuous diffusion process. This may be achieved, for instance, by allowing each token of load to follow an independent random walk on the network, as well as by applying more refined techniques admitting stronger concentration of the load distribution, cf.~\cite{DBLP:conf/focs/SauerwaldS12}. Such methods are stochastic in their very nature, and it is natural to ask whether there exist \emph{deterministic} methods which mimic this type of stochastic load balancing behavior? The answer is affirmative, with the natural candidate process being the so-called \emph{rotor-router} model.

Formally, the rotor-router mechanism is represented by an undirected anonymous graph $G = (V,E)$. Initially, a set of identical tokens is released on a vertex of the graph. At discrete, synchronous steps, the tokens are propagated according to the deterministic round robin rule, where after sending each token, the pointer is advanced to the next exit port in the fixed cyclic ordering. Such a mechanism has been proposed as a viable alternative to stochastic and random-walk-based processes in the context of load balancing problems~\cite{BerenbrinkKKMU14,DBLP:journals/cpc/CooperS06,DBLP:journals/cpc/DoerrF09}, exploration of graphs \cite{DBLP:journals/siamcomp/AfekG94,DereniowskiKPU14,Frae70,GR,KosowskiP14}, and stabilization of distributed processes \cite{BHKKR09,BhattEGT02,PhysRevLett.77.5079,YanovskiWB03}.

The resemblance between the rotor-router token distribution mechanism and stochastic balancing processes based on continuous diffusion is at least twofold, in that: (1) the number of tokens on each node for the rotor-router process has a bounded discrepancy with respect to that in the continuous diffusion process~\cite{DBLP:journals/corr/BerenbrinkKKMU14,DBLP:conf/cocoon/ShiragaYKY14}, and (2) when performing time-averaging of load over sufficiently long time intervals, the observed load averages for all nodes in the rotor-router process converge precisely to their corresponding value for the continuous diffusion process.

By contrast to time-averaged load, for any \emph{fixed} moment of time, the deterministic rotor-router process and the stochastic approaches exhibit important differences. A stochastic load balancing process based on tokens following random walks leads the system towards a ``heat death'' stochastic state, which is completely independent of the starting configuration. By contrast, the rotor-router process is a deterministic process on a graph and its limit behavior may be much more interesting.\footnote{Perhaps the first work to highlight the importance of differences between limit properties of deterministic and stochastic variants of a token-based discrete process on a graph was that of M.~Kac, in the setting of statistical mechanics (cf.~\cite{thompson2015mathematical}[Section 1.9] for a comprehensive discussion).} The number of possible configurations of a rotor-router system is finite, hence, after a transient initial phase, the process must stabilize to a cyclic sequence of states which will be repeated ever after. Natural questions arise, concerning the eventual structural behavior observed in this limit cycle of the rotor-router system, the length of the limit cycle, and the duration of the stabilization phase leading to it. So far, the only known answer concerned the case when only a single token is operating in the entire system. Yanovski \etal~\cite{YanovskiWB03} showed that such a single token stabilizes within a polynomial number of steps to periodic behavior, in which it performs a traversal of some Eulerian cycle on the directed version of the network graph.

In this work, we provide a complete structural characterization of the limit behavior of the rotor-router for an \emph{arbitrary} number $k>1$ of tokens. The obtained characterization shows that the rotor-router mechanism provides a way of self-organizing tokens, initially spread out arbitrarily over a graph, into balanced groups, each of which follows a well-defined walk in some part of the network graph. The practical implications of our result may be seen as twofold. On the one hand, when viewing the rotor-router as a load-balancing process, we obtain a better understanding of its limit behavior. On the other hand, when considering each of the tokens as a walker in the graph, we show that the rotor-router may prove to be a viable strategy for perpetual graph exploration, with possible applications in so-called network patrolling problems.

\subsection{Related Work}

\paragraph{Load balancing.}
The rotor-router mechanism of token distribution has been considered in problems of balancing workload among network nodes for specific network topologies. In this context, each token is considered as a unit-length task to be performed by one of the processors in a network of computers. Cooper and Spencer~\cite{DBLP:journals/cpc/CooperS06} studied load balancing with parallel rotor walks in $d$-dimensional grid graphs and showed a constant bound on the discrepancy between the number of tokens at a given node $v$ in the rotor-router model and the expected number of tokens at $v$ in the random-walk model. The structural properties of the distribution of tokens for a rotor-router system on the $2$-dimensional grid were considered by Doerr and Friedrich~\cite{DBLP:journals/cpc/DoerrF09}. Akbari and Berenbrink~\cite{DBLP:conf/spaa/AkbariB13} proved an upper bound of $\bigo(\log^{3/2} n)$ on the load-balancing discrepancy for hypercubes, and for tori of constant dimensions, they showed that the discrepancy is bounded by a constant. For general $d$-regular graphs, a bound of $O(d\log n /\mu)$ on the discrepancy of the rotor-router mechanism with respect to continuous diffusion follows from the general framework of~\cite{RSW98}, where $\mu$ is the eigenvalue gap of the graph, under the assumption that a sufficient number of self-loops are present at each node of the graph. This discrepancy bound has recently been improved to $O(d\sqrt{\log n /\mu})$ in~\cite{DBLP:journals/corr/BerenbrinkKKMU14}.

\paragraph{Graph exploration.}
The walks of tokens following fixed local rules at nodes provide a local mechanism of graph exploration. Each token, starting at a node of the graph, moves at each step to one of the adjacent nodes, until it has explored all the nodes of the graph. Such mechanisms tend to be location-oblivious and robust, displaying resilience to changes in network topology. In some cases, we may require the token to periodically visit all nodes, \eg, with the goal of monitoring the network or distributing updates.

One basic exploration technique relies on multiple tokens, each of which follows an independent random walk: at each step, each token chooses one of the arcs incident to the current node uniformly at random and traverses it. The performance of such parallel random walks have been analyzed by Alon~\etal~\cite{A08}, Efremenko and Reingold~\cite{DBLP:conf/approx/EfremenkoR09}, and Els\"asser and Sauerwald~\cite{DBLP:journals/tcs/ElsasserS11}, who have demonstrated that in terms of the expected time until all nodes have been visited by at least one token (i.e., the cover time), parallelization brings about a speedup of between $\Theta(\log k)$ and $\Theta(k)$ when running random walks with $k$ parallel tokens.

The rotor-router mechanism has also been studied in the context of graph exploration, sometimes under the name of {\em Edge Ant Walks\/}~\cite{Wagner99distributedcovering,YanovskiWB03}, and in the context of traversing a maze and marking edges with pebbles, \eg~in~\cite{BhattEGT02}. Cover times of rotor-router systems have been investigated by Wagner~\etal~\cite{Wagner99distributedcovering} who showed that starting from an arbitrary initial configuration\footnotemark[1], a single token following the rotor-router rule explores all nodes of a graph on $n$ nodes and $m$ edges within $\bigo(nm)$ steps. Later, Bhatt~\etal~\cite{BhattEGT02} showed that after at most $\bigo(nm)$ steps, the token continues to move periodically along an Eulerian cycle of the (directed symmetric version of the) graph. Yanovski \etal~\cite{YanovskiWB03} and Bampas~\etal~\cite{BHKKR09} studied the stabilization time and showed that the token starts circulating in the Eulerian cycle within $\Theta (mD)$ steps, in the worst case, for a graph of diameter $D$. Studies of the rotor router system for specific classes of graphs were performed in~\cite{FS10}. While all these studies were restricted to static graphs, Bampas \etal~\cite{BGKKR09}  considered the time required for the rotor-router to stabilize to a new Eulerian cycle after an edge is added or removed from the graph.\footnotetext[1]{A configuration is defined by: the cyclic order of outgoing arcs, the initial pointers at the nodes, and the current location of the token.}

\looseness-1
Studies of the parallel (\ie, multiple token) rotor-router were performed by Yanovski \etal~\cite{YanovskiWB03} and Klasing~\etal~\cite{DBLP:conf/podc/KlasingKPS13}, and the speedup of the system due to parallelization was considered for both worst-case and best-case scenarios.  In \cite{DereniowskiKPU14}, Dereniowski~\etal\ establish bounds on the minimum and maximum possible cover time for a worst-case initialization of a $k$-rotor-router system in a graph $G$ with $m$ edges and diameter $D$,  as  $\Omega(mD/k)$ and $\bigo(mD/\log k)$ respectively.  In \cite{KosowskiP14}, Kosowski and Pająk provided a more detailed analysis of the speedup for specific classes of graphs, providing tight bounds of cover-time speed-up for all values of $k$ for degree-restricted expanders, random graphs, and constant-dimensional tori. For hypercubes, they resolve the question precisely, except for values of $k$ much larger than $n$.


\subsection{Our Results}

In this work we provide a structural characterization of the limit behavior of the rotor-router model with multiple tokens. Yanovski \etal~\cite{YanovskiWB03} conjectured that the rotor-router system enters a short sequence of states (of length at most $2m$), which repeats cyclically ever after. We start this work by disproving this conjecture. In fact, we display an example of a starting configuration which admits a limit cycle with a period of superpolynomial length ($\text{exp}(\Omega(\sqrt{n \log n}))$) with respect to the size of the graph. Our example is similar to the construction presented by Kiwi~\etal~\cite{Kiwi94nopolynomial} to prove the existence of super-polynomial periods for chip firing games on graphs (although the rules of chip firing games are only very loosely related to those of the rotor-router).

By contrast, it turns out the fact that the rotor-router admits long limit cycles does not signify that the limit behavior of the rotor-router should be perceived as a ``disordered'' discrete dynamical system. The long period in our counterexample comes from the system being composed from many smaller parts, each of which exhibits a small (but different) period length. 
We show that for any limit sequence of states in the rotor-router model, the graph can be partitioned into arc-disjoint directed Eulerian cycles, with each token in the limit periodically traversing arcs of one particular cycle. We name such behavior a \emph{subcycle decomposition}, the exact properties of which are described in Section~\ref{sec:lock-in}. To complement the lower bound, we provide an upper bound of $\text{exp}(\bigo(\sqrt{m \log m}))$ on the period of parallel rotor walks in its limit behavior. This upper bound asymptotically almost matches the lower bound from our example.

There are several consequences of our structural characterization of the limit behavior of the rotor-router.
First, we show that it is possible to determine efficiently whether the system has already stabilized (\ie, reached a configuration that will repeat itself) or not. This detection is based on the analysis of the properties of stable states, that is, of how the tokens arriving at a node are distributed into groups leaving on different outgoing arcs. The main point of this analysis is the observation that the cumulative number of tokens entering a vertex $v$ (over the time period $\{t,(t+1),\ldots,(t+\dt)\}$) is equal to the cumulative number of tokens leaving vertex $v$ (over time $\{(t+1),(t+2),\ldots,(t+\dt+1)\}$).

Next, by defining an appropriate potential of a system and showing its monotonicity, we can give a polynomial bound on a number of steps necessary for a system with an arbitrary initialization to reach a periodic configuration. We provide an upper bound of $\bigo(m^4D^2 + mD\log k)$, together with examples of graphs with initial configuration having just 2 tokens that require $\Omega(m^2 \log n)$ steps. This analysis is presented in Section~\ref{sec:stab}. The obtained polynomial upper bound means that the rotor-router is an efficient means of self-organizing tokens so as to perform a periodic traversal of the edges of the graph.

Finally, Section~\ref{sec:simulation} is dedicated to showing how the previous results can be applied in a constructive way with regard to efficient simulation of a rotor-router system.  We show how the properties of subcycle decomposition can be applied to provide a way to preprocess any starting configuration in a way that makes it possible to answer queries of certain type in a polynomial time. This shows that a structural characterization of the rotor-router system is not only an important as a theoretical tool for understanding the limit behavior of the system, but it also as a practical tool for solving certain problems related to the rotor-router system.

As a complementary result, we show for the single-token rotor-router how
to efficiently compute the Eulerian traversal cycle on which the token would be locked-in, faster than by running the process directly. A naive simulation would take $\bigo(mD)$ time, but by using the structural properties of a single token walk together with application of efficient data structures we show how to preprocess the input graph in time $\bigo(n+m)$ such that we can answer queries about token position at any given time $T$, in $\bigo(\log \log m)$ time per query.


\vspace{-0.2cm}

\section{Model and Preliminaries}

Let $G=(V,E)$ be an undirected connected graph with $n$ nodes, $m$ edges and \emph{diameter} $D$. Let $k$ be the number of tokens. The digraph $\vec{G} = (V,\vec{E})$ is the directed version of $G$ created by replacing every edge $(u,v)$ with two directed arcs $\vec{uv}$ and $\vec{vu}$. We will refer to the undirected links in graph $G$ as \emph{edges} and to the directed links in the graph $\vec{G}$ as \emph{arcs}. Given a vertex $v$, we will denote its set of incoming arcs by $\inedg(v)$ and outgoing arcs by $\outedg(v)$.
Each vertex $v$ of $G$ is equipped with a fixed ordering of all its outgoing arcs $\rho_v = (e_1,e_2, \ldots, e_{\text{deg}(v)})$.

The precise definition of the rotor-router model on the system defined by $(\vec{G},(\rho_v)_{v\in V})$ is as follows:\\ A \emph{state} at the current time step  $t$ is a tuple:
$\state_t = ((\pointer_v)_{v \in V},(\tokens_v)_{v \in V}),$
where $\pointer_v$ is an arc outgoing from node $v$, which is referred to as \emph{the current port pointer at node} $v$, and $\tokens_v$ is the number of tokens at any given node.  For an arc $\vec{(vu)}$, let $\mathit{next}\vec{(vu)}$ denote the arc after the arc $\vec{(vu)}$ in the cyclic order $\rho_v$.
%
During each step, each node $v$ distributes in round-robin fashion all of its tokens, using the following algorithm:\\[2mm]
\enlargethispage{5mm}
While there is a token at node $v$, do
\begin{enumerate}
\item Send token to $\pointer_v$,
\item Set $\pointer_v = \mathit{next}(\pointer_v)$.
\end{enumerate}
Note that during a single time step all tokens at a node $v$ are sent out and at exactly the next time step all those tokens arrive at their respective destination nodes.

For a given state $\state_t$, we say that it is \emph{stable} iff there exists $t' > t$ such that $\state_{t'} = \state_t$.
\emph{The stabilization time} of state $\state_0$, denoted $\stab$, is the smallest value such that $\state_{\stab}$ is stable.
We call \emph{the periodicity} of state $\state_0$ the smallest $\per>0$ such that $\state_{\stab} = \state_{\stab+\per}$.

Throughout the paper, we denote multisets using $\mset{}$ notation, while for integer ranges, we write $\closedrange{a}{b} \defeq \{a,a+1,\ldots,b\},$ $\halfrange{a}{b} \defeq \{a,a+1,\ldots,b-1\}.$

\section{Periodicity of the Rotor-Router System}
\label{sec:lock-in}


We begin with the observation that knowledge of the first $\stab+\per$ states of the system, that is $\state_0,\ldots,\state_{\stab+\per-1}$, gives us full knowledge of any future state for arbitrarily large time $t\ge \stab$:
$\state_t = \state_{\stab + ((t - \stab) \bmod \per)}.$

So as to be able to efficiently predict the future evolution of any rotor-router state, it would be useful to put a polynomial bound on $\per$ and $\stab$ (with respect to $n,m$ and $k$). If $k=1$, due to results from Yanovski~\etal~\cite{YanovskiWB03}
, we know that $\per = 2m$ and $\stab = O(mD)$.
For arbitrary $k$, Yanovski~\etal~\cite{YanovskiWB03} conjectured that $\per \le 2m$ for any graph $G$ regardless of the initial state.
However, the following negative result disproves their conjecture and shows that the periodicity cannot be polynomially bounded for parallel rotor-routers.

\begin{theorem}
\label{th:exponential}
There exists a family of graphs and initial states, with $k=2m$ tokens, having the periodicity $\per = 2^{\Omega(\sqrt{n \log n})}$.
\end{theorem}

\begin{proof}
We will construct such a family of graphs $\mathcal{G}_r$ for any sufficiently large integer $r$, and an appropriate initial configuration of tokens. First consider a balloon graph $G_x$ consisting of a cycle of $x > 3$ vertices $\{v_0,v_1,\dots v_{x-1}\}$ and an additional vertex $v_x$ (called the \emph{base vertex}) that is joined by an edge to vertex $v_{x-1}$ of the cycle (see Figure~\ref{fig:expgraph}(a)). Let the initial token distribution at vertices $(v_0,\dots, v_x)$ be $(1,2,2,\dots 2,4,1)$. Further let the exit pointers at vertex $v_i$, $ 0\leq i \leq x-2$ be oriented towards $v_{i-1 \mod{x}}$ (in the counter-clockwise direction, in the figure), while at the vertex $v_{x-1}$ the exit pointer is oriented towards $v_0$ (i.e. in the opposite direction). At the base vertex $v_x$ there is only one outgoing arc and so, the exit pointer at $v_x$ will always point towards this arc.

Observe that for a vertex of out-degree two, the exit pointer remains unchanged if an even number of tokens exit this vertex in the current round, while the exit pointer is rotated if an odd number of tokens exit in the current round.

We will now analyze the movement of tokens along the arcs of the graph in each round. Figure~\ref{fig:Periodballoon}  shows an example for a balloon graph $G_x$, where $x=5$. During the first round, the number of tokens moving on the arcs $(v_0,v_1)$, $(v_1,v_2)$, $\dots (v_{x-1}, v_0)$ of the cycle (in the clockwise direction) is given by the sequence $S_0$=$(0,1,1,\dots,1,2)$. During the same round, the number of tokens moving on the arcs in the counter-clockwise direction on the cycle is given by $(1,1,\dots 1)$. On the branch edge $(v_{x-1}, v_x)$ there is exactly one token moving in each direction (See Figure~\ref{fig:Periodballoon}(a)).

\begin{figure}[!t]
\centering\includegraphics[width=0.8\textwidth]{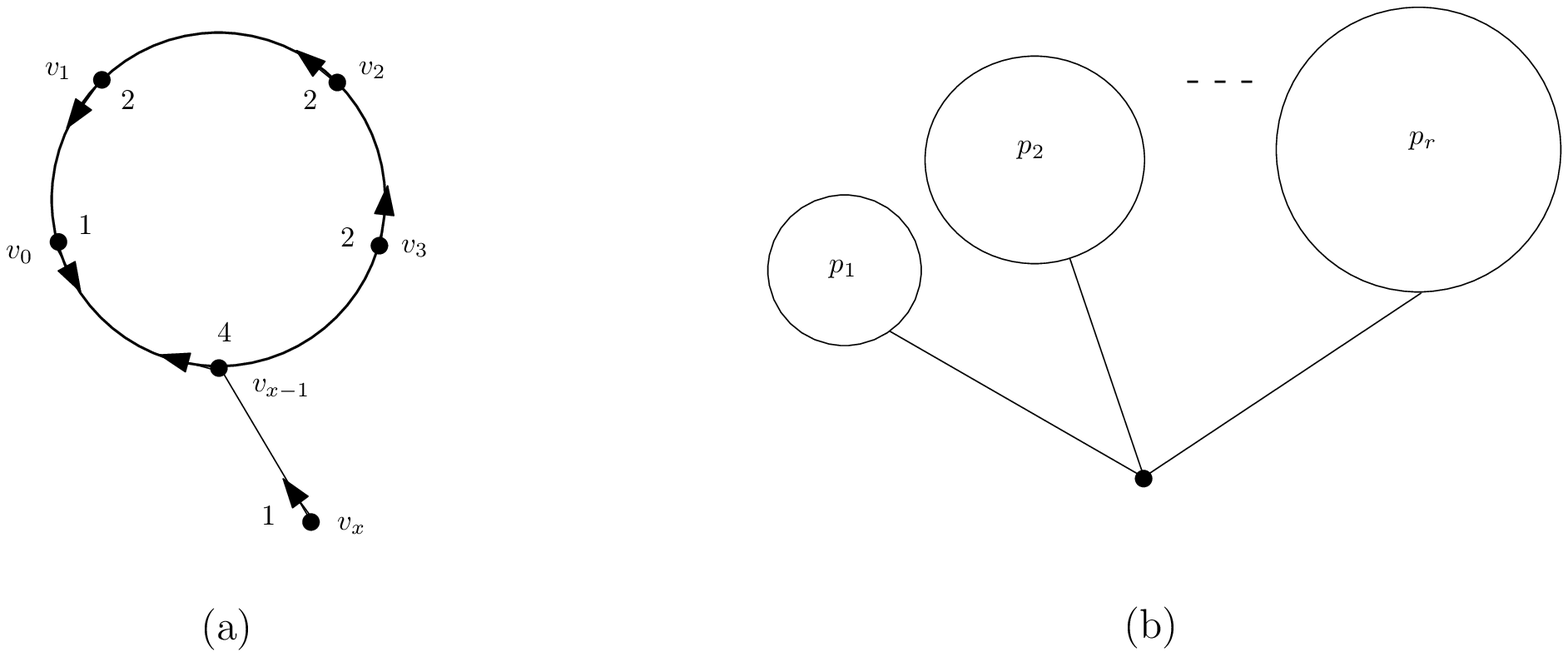}
\caption{ \small{(a) The balloon graph and the initial token distribution. (b) The family of graphs $\mathcal{G}_r$ consisting of $r$ balloons.}
} \label{fig:expgraph}
\end{figure}

\begin{figure}[!h]
\includegraphics[width=\textwidth]{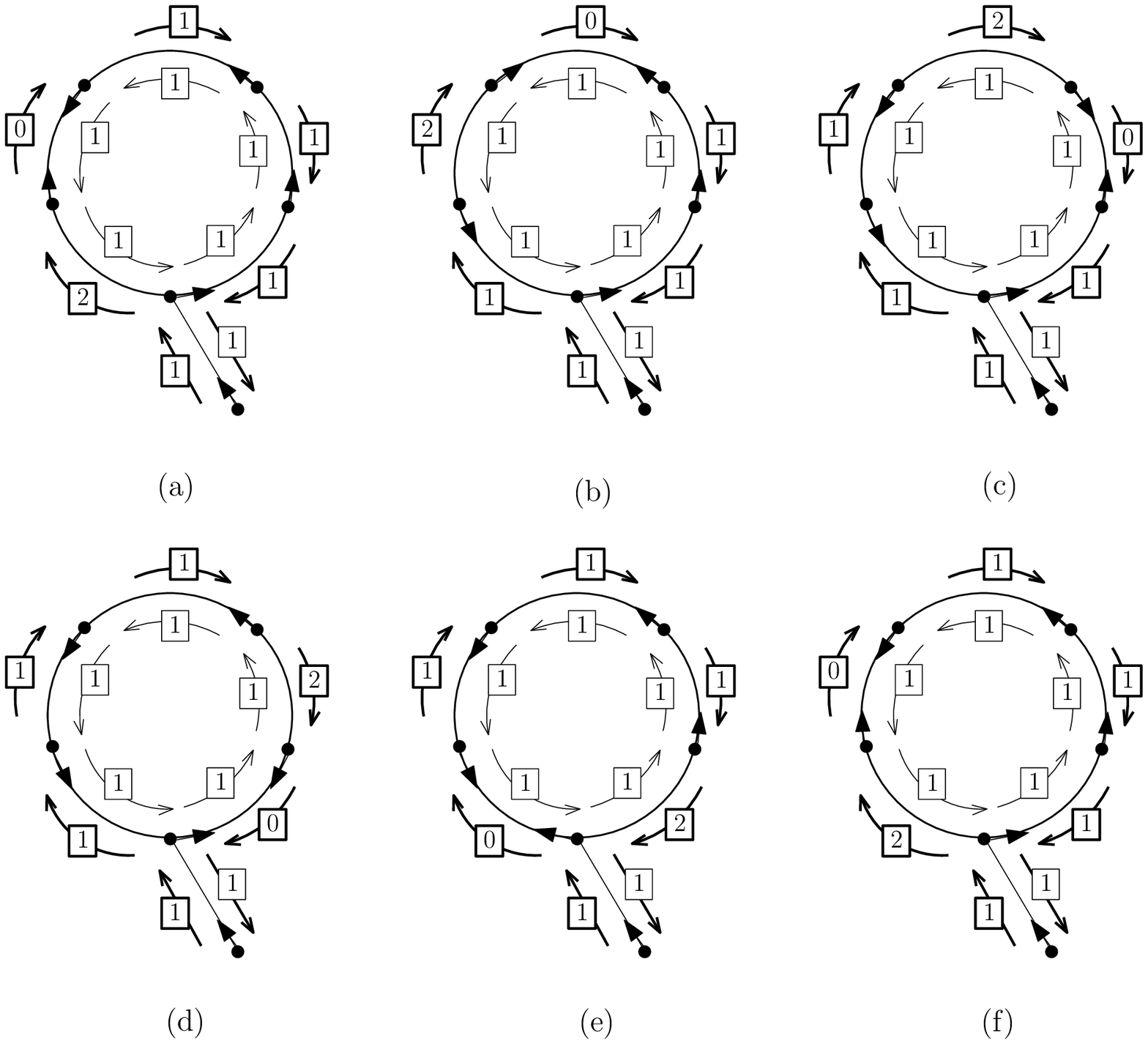}
\caption{ \small{The number of tokens circulating on the arcs of a balloon graph $G_5$ in each round.}
} \label{fig:Periodballoon}
\end{figure}

During the second round, the number of tokens moving on the arcs $(v_0,v_1)$, $(v_1,v_2)$, $\dots (v_{x-1}, v_0)$ of the cycle (in the clockwise direction) is given by the sequence\\$S_1$=$(1,1,\dots,1,2,0)$ which is a cyclic rotation of the sequence $S_0$. The number of tokens moving on the arcs in the counter-clockwise direction on the cycle is still given by $(1,1,\dots 1)$. Again, the branch edge $(v_{x-1}, v_x)$ has exactly one token moving in each direction (See Figure~\ref{fig:Periodballoon}(b)).

Continuing with the above analysis, it is easy to see that in subsequent rounds, the number of tokens moving on the arcs of the cycle (in the clockwise direction) is given by cyclic rotations of $S_0$, \ie, by the sequences $(1,\dots,1,2,0,1)$, $(1,\dots,1,2,0,1,1)$, $(1,\dots,1,2,0,1,1,1)$ and so on (See Figure~\ref{fig:Periodballoon}(c-e)). The number of tokens moving along the cycle in the counterclockwise direction is always one token per arc of the cycle. On the branch edge $(v_0,v_x)$ there is exactly one token moving in each direction in each round. Since the length of the sequence $S_0$ is $|S_0|=x$, after every $x$ steps the configuration of tokens moving on the arcs of the cycle is the same. In other words,
the periodicity of this rotor-router system is $x$. Notice that the graph $G_x$ has $x+1$ vertices and $2(x+1)$ arcs, and there are exactly $2(x+1)$ tokens in the system.

We will now construct the family of graphs $\mathcal{G}_r$. For any given $r$, let $p_1,p_2,\dots p_r$ be the first $r$ prime numbers starting from $p_1=3$. We take $r$ balloon graphs of sizes  $(1+p_1),(1+p_2),\dots,(1+p_r)$ respectively and join them by merging all the base vertices into one vertex (see Figure~\ref{fig:expgraph}(b)). In each balloon graph we place the tokens as before, such that the merged base vertex now contains $r$ tokens. During each step, $r$ tokens will exit the base vertex through the $r$ outgoing arcs and $r$ other tokens will enter the base vertex through the $r$ incoming arcs. Thus, irrespective of the initial state of the exit pointer at the base vertex, the system will behave in the same manner. The behavior of the system in the distinct balloons would be independent of each other and for each balloon of size $(1+p_i)$ the configuration of the balloon would repeat itself in exactly $p_i$ steps as before. Thus, the global state of the system would repeat in $\text{lcm}(p_1,\dots p_r)$ = $\Pi_{i=1}^{r}p_i$ steps. Note that the size of the graph, $\mathcal{G}_r$, is given by $n=1+\sum_{i=1}^{r}p_i = \Theta(r^2 \log{r})$. In general, for any given integer $n$, we can construct a similar example graph by partitioning the $n-1$ vertices into balloons of appropriate sizes joined to the $n$th vertex, such that the period of the system is equal to the \emph{Landau function} \cite{Landau} $g(n-1)$ $= 2^{\Omega(\sqrt{n\log{n}})}$.

\end{proof}

We remark that a similar result exists for parallel chip-firing games~\cite{Kiwi94nopolynomial}.


We now present an upper bound on the periodicity of $k$ parallel rotor walks, for arbitrary values of $k$.
First, we will show that even though a stable state can exhibit very long (super-polynomial) periodicity, the underlying graph $G$ can be partitioned into parts, such that each part separately exhibits small (linear) periodicity.

We will use calligraphic large letters (\eg $\load : \vec{E} \cup V \to \mathbb{Z}$) to denote \emph{token distributions}. Thus:
\begin{itemize}
\item $\load_t(v)$ (\emph{load} of node $v$) is number of tokens located at node $v$ in time step $t$,
\item $\load_t(e)$ (\emph{load} of arc $e$) is number of tokens sent out on arc $e$ at time step $t$.
\end{itemize}
Thus, although tokens cannot be located on edges in our model, we can view the tokens in vertex as located already on the outgoing ports that they will be distributed to.

We will use specifically $\load_t$ to denote token distribution associated with state $\state_t$.
(It is important to note, that it is possible for two states to satisfy $\forall_{e} \load_t(e) = \load_{t'}(e)$ and yet $\state_t \not= \state_{t'}$, as we also require that pointers be in the same positions in identical states.)

We begin with a series of observations on the token distribution process in a rotor-router system.

\begin{observation}
\label{ob:load1}
Since every token is pushed onto some outgoing arc, we have:\\
$\sum_{e \in \outedg(v)} \load_t(e) = \load_{t}(v);
\sum_{e \in \inedg(v)} \load_t(e) = \load_{t+1}(v).$
\end{observation}

We also generalize token distribution into the \emph{cumulative token distribution}. Given two time steps $t_1 \le t_2$, we define
$\cload_{t_1}^{t_2} \defeq \sum_{t \in \halfrange{t_1}{t_2}} \load_t,$
in particular, for a vertex $v$ and arc $e$:
$$\cload_{t_1}^{t_2}(v) \defeq \sum_{t \in \halfrange{t_1}{t_2}} \load_t(v),\ \cload_{t_1}^{t_2}(e) \defeq \sum_{t \in \halfrange{t_1}{t_2}} \load_t(e).$$

Consequently, a natural generalization of Observation~\ref{ob:load1} from load to cumulative load is as follow:

\begin{observation}
\label{ob:cload1}
$\sum_{e \in \outedg(v)} \cload_{t_1}^{t_2}(e) = \cload_{t_1}^{t_2}(v);
\sum_{e \in \inedg(v)} \cload_{t_1}^{t_2}(e) = \cload_{t_1+1}^{t_2+1}(v).$
\end{observation}

The next observation follows from the fact that rotor-router distributes tokens in a round-robin fashion among all outgoing arcs of a vertex:
\begin{observation}
\label{ob:load2}
$\forall_{e_1,e_2 \in \outedg(v)}  |\cload_{t_1}^{t_2}(e_1) - \cload_{t_1}^{t_2}(e_2)| \le 1.$
\end{observation}
Since arbitrarily large discrepancies (between two incoming edges) of incoming number of tokens are smoothed discretely, we can see the process of token propagation as a load-balancing scheme.

We now define the concept of \emph{potential} of a token distribution system, which will be helpful to derive the necessary and sufficient conditions for a system state to be stable.

\begin{definition}
\label{def:pot}
Given a token distribution $\mathcal{A}$ over edges, we define its potential as:
$\Phi(\mathcal{A}) \defeq \sum_{e \in \vec{E}} \left(\mathcal{A}(e)\right)^2.$\\
We also introduce a shorthand notation for the $i$-th potential of a given rotor-router state $\state_t$ as:
$\Phi_i(\state_t) \defeq \Phi(\cload_t^{t+i}) = \sum_{e \in \vec{E}} \left(\cload_{t}^{t+i}(e)\right)^2.$
\end{definition}
Note that $\Phi_1 \equiv \Phi$. It is important to note that while arbitrary convex function can be used in the potential definition, usage of quadratic function will prove advantageous when analyzing the speed of convergence to a stable state, not only its properties.

The following folklore lemma provides us with a characterization of the minimum of the potential sums.
\begin{lemma}
\label{lem:sum_of_sq}
Over all partitions of integer $S$ into $d$ integers, the partition $\mset{\lfloor\frac{S}{d}\rfloor,\ldots,\lfloor\frac{S}{d}\rfloor,\allowbreak\lceil\frac{S}{d}\rceil,\ldots,\lceil\frac{S}{d}\rceil}$ uniquely minimizes the value of sum of squares of elements.
\end{lemma}
\emph{All omitted proofs are provided in the Appendix.}

\begin{lemma}
\label{lem:pot_drop}
For arbitrary $i$ and $t$, the $i$-th potential is non-increasing:
$\Phi_{i}(\state_{t+1}) \le \Phi_{i}(\state_{t}).$
\end{lemma}

\begin{proof}
To prove the lemma we have to observe how the round-robin property of the rotor-router acts locally on the groups of tokens (cumulative over the time interval $[t,t+i)$).
From Observation~\ref{ob:cload1} we know, that:
$$\sum_{e \in \outedg(v)} \cload_{t+1}^{t+i+1}(e) = \cload_{t+1}^{t+i+1}(v) = \sum_{e \in \inedg(v)} \cload_{t}^{t+i}(e).$$
However, from Observation~\ref{ob:load2} and Lemma~\ref{lem:sum_of_sq} we get that the multiset of values over outgoing arcs \emph{minimizes} the sum of squares.
Thus:
$$\sum_{e \in \outedg(v)} (\cload_{t+1}^{t+i+1}(e))^2 \le \sum_{e \in \inedg(v)} (\cload_{t}^{t+i}(e))^2,$$
which leads to:
$$\Phi_{i}(\state_{t+1}) = \sum_v \sum_{e \in \outedg(v)} (\cload_{t+1}^{t+i+1}(e))^2 \le \sum_v \sum_{e \in \inedg(v)} (\cload_{t}^{t+i}(e))^2 = \Phi_{i}(\state_{t}).$$

\end{proof}

Observe that Lemma~\ref{lem:pot_drop} implies that if the system is stable, all of the potentials are preserved at every (future) time step. This observation is powerful enough to derive strong characterization of stable states (see Theorem~\ref{th:subcycle}, equivalence of \eqref{subcycle1} and \eqref{subcycle2}). However, in order to be able to reason about bounds on stabilization time, we need a more powerful notion of being able to characterize even the temporary regularities in token trajectories (for not necessarily stable states).

\begin{definition}
\label{def:dec}
We say that a state $\state_{T}$ admits a \emph{$\dt$-step subcycle decomposition}, if in every vertex $v$ we can define a one-to-one mapping between incoming and outgoing arcs of $v$ $M_v : \inedg(v) \to \outedg(v)$, such that:
\begin{equation}
\label{eq:dec}
\forall_{e \in \inedg(v)}  \forall_{t \in \halfrange{T}{T+\dt}} \load_{t}(e) = \load_{t+1}(M_v(e)). \end{equation}
\end{definition}
The subcycle decomposition has the following equivalent interpretation. We partition $\vec{E} = \vec{E_1} \cup \ldots \cup \vec{E_c}$, such that each $\vec{E_i}$ induces a strongly-connected subgraph of $G$, and for each $\vec{E_i}$ there exists an Eulerian cycle covering it such that each token traversing arcs of $\vec{E_i}$ follows this particular Eulerian cycle during time steps $T,T+1,\ldots,T+\dt-1$.

Observe that the mapping $M$ in Definition~\ref{def:dec} does not need to be necessarily unique. We will call any such mapping $M$ \emph{a valid mapping with respect to $\dt$ subcycle decomposition} if \eqref{eq:dec} holds for it.

The following lemma gives a series of equivalent characterizations of subcycle decomposition, connecting the existence of such a decomposition during any time interval with lack of potential drop during the time interval, as well as a load-balancing discrepancy criterion over all shorter sub-intervals of time.
\begin{lemma}
\label{lem:long}
The following statements are equivalent:
\begin{enumerate}[(i)]
\item \label{lem:long:item1} $\state_{T}$ admits a $\dt$-step subcycle decomposition,
\item \label{lem:long:item2} $\forall_v \forall_{ t,t' : \halfrange{t}{t'} \subseteq \halfrange{T}{T+\dt} }, \mset{\cload_{t}^{t'}(e)}_{e\in \inedg(v)} = \mset{\cload_{t+1}^{t'+1}(e)}_{e\in \outedg(v)}$ (multisets of cumulative loads are preserved locally),
\item \label{lem:long:item3} $\forall_{  t,t' : \halfrange{t}{t'} \subseteq \halfrange{T}{T+\dt} } \mset{\cload_{t}^{t'}(e)}_{e \in \vec{E}}  = \mset{ \cload_{t+1}^{t'+1}(e)}_{e \in \vec{E}}$ (multisets of cumulative loads are preserved globally),
\item \label{lem:long:item4} $\forall_{0 \le i \le \dt} \Phi_i(\state_{T}) = \Phi_i(\state_{T+1}) = \ldots = \Phi_i(\state_{T+\dt-i+1})$ (potential is constant),
\item \label{lem:long:item5} $\forall_v \forall_{e_1,e_2 \in \inedg(v)} \forall_{  t,t' : \halfrange{t}{t'} \subseteq \halfrange{T}{T+\dt} }  |\cload_{t}^{t'}(e_1) - \cload_{t}^{t'}(e_2)| \le 1$ (incoming discrepancy is at most one).
\end{enumerate}
\end{lemma}

For a fixed value of $\dt=1$, Lemma~\ref{lem:long} captures the property that as long as $\Phi(\state_{T})$ remains constant, the loads on edges are only permuted between any two consecutive timesteps. Coupled with Lemma~\ref{lem:pot_drop} it immediately implies aforementioned property. Unfortunately, this property is not strong enough for our purposes. However, for arbitrary values of $\dt$, we can still employ the notion of higher order potentials $\Phi_{\dt}$ (as defined previously), and observe the load balancing properties from Observation~\ref{ob:cload1}. The fact that stable state, by Lemma~\ref{lem:pot_drop} and Lemma~\ref{lem:long} admits load balancing properties even when collapsing multiple timesteps into a single frame, is our lever which will be used to derive strong properties of such states in the rest of this section.

\begin{definition}
We say that a state $\state_t$ admits a \emph{$\infty$-subcycle decomposition} if it admits $i$-steps subcycle decomposition for arbitrarily large i.
\end{definition}

Now we proceed to obtain a more algorithmic characterization of stable states. First, we show that if we do not experience a potential drop during $2m^2$ time steps, then the rotor-router system has reached its limit configuration.

\begin{theorem}
\label{th:subcycle}
The following conditions are equivalent:
\begin{enumerate}[(i)]
\item \label{subcycle1} $\state_T$ is stable,
\item \label{subcycle2} $\state_T$ admits a $\infty$-subcycle decomposition,
\item \label{subcycle3} $\state_T$ admits a $(2m^2)$-subcycle decomposition.
\end{enumerate}
\end{theorem}

As a direct consequence of the proof of Theorem~\ref{th:subcycle}, we have:
\begin{corollary}
\label{ob:subc}
For a stable state $\state_T$, any mapping between incoming and outgoing arcs of $v$ denoted $M_v : \inedg(v) \to \outedg(v)$ that is valid with respect to $2m^2$-subcycle decomposition, is also valid with respect to $\infty$-subcycle decomposition.
\end{corollary}

We are now ready to provide a stronger characterization of a stable state in Theorem~\ref{th:stable_charact} (compared to Theorem~\ref{th:subcycle}), based on refined analysis of the potential behavior.
\begin{theorem}
\label{th:stable_charact}
State $\state_T$ is stable iff:
$\sum_{i=1}^{3m} \Phi_{i}(\state_T) = \sum_{i=1}^{3m} \Phi_{i}(\state_{T+2m^2}).$
\end{theorem}

Finally, we provide an upper bound on the length of the period for any rotor-router state. It is interesting to see that the upper bound is not far from the period of the example graph given in Theorem~\ref{th:exponential}.

\begin{theorem}
For any stable state $S_{T}$, the period length of the limit cycle is bounded by $t_p = \bigo(\exp(\sqrt{(m \log m)}))$.
\end{theorem}

\begin{proof}
Observe, that any arc in $G$ can be part of exactly one cycle in any given subcycle decomposition. As the period length of any stable state is upperbounded by the least common multiple of the length of the cycles, we get the desired upper bound as the value of the \emph{Landau's function} on the total number of arcs in $G$.
\end{proof}

\section{Stabilization Time of the Rotor-Router System}
\label{sec:stab}
In this section we provide upper and lower bounds on the stabilization time of parallel rotor-router systems.
Since the values of potentials are discrete and non-increasing, in Theorem~\ref{th:stable_charact} we have a very powerful tool to reason about the stabilization of a state --- if the sum of potentials remains unchanged for more than $3m$ time steps, the system has reached a stable state. Thus, we can naively bound each of the potentials by $\bigo((mk)^2)$, and so also bound the sum of potentials by $\bigo(m \cdot (mk)^2)$. This gives the following corollary.

\begin{corollary}
For any initial state $\state_0$, there exists $T = O(m^5k^2)$ such that $\state_T$ is stable.
\end{corollary}

We will now show how to obtain an even stronger bound (in terms of dependence on $k$), but for this we need a refined upper bound on initial potential. To achieve this, we need to treat the rotor-router system as a load balancing process.

\paragraph{Round-fair processes.} As we intend to provide bounds on the values of the $i$-th potential for rotor-router process (given sufficiently long initialization time), we need to analyze the behavior of the cumulative rotor-router processes, \ie, for a fixed $\dt$, to observe how the distribution of tokens $\cload_{t}^{t+\dt}$ evolves with time. Thus, in the following, we will use the broader concept of \emph{round-fair processes} denoted by $\mathcal{W}$, as introduced in~~\cite{RSW98}.  Specifically, we will call an algorithm \emph{strictly fair} if, in every step, the number of tokens that are sent out over any two edges incident to a node differs by at most one.

\begin{definition}
A process of token distribution (denoted by $\mathcal{W}$) is \emph{round-fair},
if:
\begin{equation}
\label{eq:roundfair}
\forall_{e\in \outedg(v)} \wload_t(e) \in \left\{\left\lfloor \frac{\wload_t(v)}{\deg(v)} \right\rfloor,\left\lceil\frac{\wload_t(v)}{\deg(v)}\right\rceil \right\}
\end{equation}
and no tokens are left in nodes:
$ \wload_t(v) = \sum_{e \in \outedg(v)} \wload_t(e).$
\end{definition}
We observe that any rotor-router process is round-fair. Also, by Observations~\ref{ob:load1}, \ref{ob:cload1} and \ref{ob:load2}, for any fixed $\dt$, cumulative rotor-router in the sense of $\wload_t = \cload_t^{t+\dt}$ is also round-fair.

The round-fairness condition can be strengthened into algorithms which are \emph{cumulatively fair}. We will call an algorithm \emph{cumulatively fair} if for every interval of consecutive time steps, the total number of tokens sent out by a node differs by at most a small constant for any two adjacent edges.
It is easy to see that cumulative fair algorithms under the constraint that every token is propagated, are performing the rotor-router distribution (and vice versa, rotor-router distribution is cumulative fair with every token propagated).

In the rest of this section, in order to simplify notation, we will assume that $G$ is not bipartite. For the full formulation of Definition~\ref{def:distance}, Lemma~\ref{lem:even_length}, $\mathcal{B}(\mathcal{W}_t)$ and Definition~\ref{def:discr}, refer to the Appendix~\ref{appendix:sec4}.

\begin{definition}
A sequence of arcs $e_1,e_2,\ldots,e_p$ is called an \emph{alternating path} (of length $p-1$), if either every pair of arcs $e_{2i},e_{2i+1}$ shares staring vertex and every pair of arcs $e_{2i-1},e_{2i}$ shares ending vertex, or vice-versa: every pair of arcs $e_{2i},e_{2i+1}$ shares ending vertex and every pair of arcs $e_{2i-1},e_{2i}$ shares starting vertex.
\end{definition}

\begin{definition}
\label{def:distance}
We define the notion of \emph{distance} between two arcs $e,e'$, denoted by $d(e,e')$, as a length of shortest alternating path having $e$ and $e'$ as first and last arcs.
\end{definition}
In other words, a distance can be treated as a transitive closure of a relation where we define any pair of arcs sharing starting or ending vertex as at distance one.

\begin{lemma}
\label{lem:even_length}
For any two arcs $e,e'$ in non-bipartite graph $G$:
$d(e,e') \le 4D+1$
moreover there exists an alternating path connecting $e$ and $e'$ containing at most $2D$ pairs of arcs sharing ending vertices and at most $2D+1$ pairs of arcs sharing starting vertices.
\end{lemma}

Now we will proceed to analyze the behavior of the potential defined as in Definition~\ref{def:pot},
with respect to a round-fair processes.

Recall that $\Phi(\wload_t) \defeq \sum_{e \in \vec{E}} \left(\wload_t(e)\right)^2.$
We will denote the smallest value of the potential achieved by distribution of tokens that preserves sums of loads over load balancing sets of arcs (ignoring the restriction that loads are integers) by:

\begin{equation}\label{eq:b}\mathcal{B}(\wload_t) \defeq 2m\cdot\left( \avg_{e \in \vec{E}} \wload_t(e)\right)^{\!2},\end{equation}

For non-bipartite graphs, \eqref{eq:b} reduces to the following form: $\mathcal{B}(\wload_t) = \frac{k^2}{2m} = \const.$ The following lemma follows directly from the convexity of quadratic functions.

\begin{lemma}
$\Phi(\wload_t) \ge \mathcal{B}(\wload_t).$
\end{lemma}

\begin{definition}
\label{def:discr}
We say that a configuration of tokens $\wload_t$ in non-bipartite graph $G$ has discrepancy over arcs equal to $\max_{e,e' \in \vec{E}} (\wload_t(e) - \wload_t(e'))$.
\end{definition}

The next observation follows directly from \eqref{eq:roundfair}.
\begin{observation}
The discrepancy over arcs is non-increasing in time, that is:
$$\max_{e,e' \in \vec{E}} (\wload_t(e) - \wload_t(e')) \ge \max_{e,e' \in \vec{E}} (\wload_{t+1}(e) - \wload_{t+1}(e')).$$
\end{observation}

We also put the following bound on the potential drop with respect to the discrepancy of number of tokens over arcs.
\begin{lemma}
\label{lem:potdrop}
Consider a timestep $t$ such that $\mathcal{W}_t$ has discrepancy over arcs $x>4D+1$.
Then:
$\Phi(\wload_t) - \Phi(\wload_{t+1}) \ge \frac{(x-4D-1)(x-1)}{4D}.$
\end{lemma}

\begin{lemma}
\label{lem:atleastdiscrp}
If $\wload_t$ has discrepancy $x$, then
$\Phi(\wload_t) \le \mathcal{B}(\wload_0) + \frac12 m x^2$.

\end{lemma}

\begin{theorem}
\label{th:discrepancy}
If $T \ge 16mD \ln k$, then $\wload_T$ has discrepancy over arcs at most $10D$.
\end{theorem}
\begin{proof}
Observe that for discrepancies $x \ge 10D$ we have, by Lemma~\ref{lem:potdrop}:
$$\Phi(\wload_t) - \Phi(\wload_{t+1}) \ge \frac{(x-4D-1)(x-1)}{4D} \ge \frac{x^2}{16D}.$$
However, by Lemma~\ref{lem:atleastdiscrp}:
$x^2 \ge 2\frac{(\Phi(\wload_t)-\mathcal{B}(\wload_0))}{m}.$
Thus:
$$\Phi(\wload_{t+1})-\mathcal{B}(\wload_0) \le (\Phi(\wload_t)-\mathcal{B}(\wload_0))\left(1-\frac{1}{8mD}\right).$$

Let us assume that after $T \ge 16mD \ln k$ steps the discrepancy is larger than $10D$.
Since
$\Phi(\wload_0)-\mathcal{B}(\wload_0) \le \Phi(\wload_0) \le k^2,$
we have:
$$\Phi(\wload_T)-\mathcal{B}(\wload_0) \le k^2 \cdot \left(1-\frac{1}{8mD}\right)^{16mD\ln k} < k^2 (1/e)^{2 \ln k}  = 1,$$
implying that $\Phi(\wload_T)= \lceil \mathcal{B}(\wload_0) \rceil$, which implies that $\Phi(\wload_T)$ minimizes potential among integer load distribution. Thus $\wload_T$ has discrepancy at most 1, a contradiction.

\end{proof}

We are now ready to prove our main result on the time of stabilization of any rotor-router initial state.

\begin{theorem}
\label{th:stabilizationtime}
For any initial state $\state_0$, there exists $T = O(m^4D^2 + mD\log k)$ such that $\state_T$ is stable.
\end{theorem}
\begin{proof}
Let $t_0 = \left\lceil16mD \ln(3km)\right\rceil$.
We observe that the cumulative rotor-router process (taken over $\dt$ rounds) is round-fair, with the number of tokens equal to $\dt \cdot k$.
For $t \ge t_0$, $\dt \le 3m$, by Theorem~\ref{th:discrepancy} the token distribution  of $\cload_t^{t+\dt}$ has discrepancy over arcs at most $10D$, thus:
$\mathcal{B}( \cload_0^{\dt} ) \le \Phi_{\dt}(\state_t) = \Phi( \cload_t^{t+\dt} ) \stackrel{\eqref{lem:atleastdiscrp}}{\le} \mathcal{B}( \cload_0^{\dt} ) + 50m D^2.$
We next obtain:
\begin{equation}
\label{eq:sum_of_pot}
 \sum_{i=1}^{3m}\mathcal{B}( \cload_0^{i} )  \le \sum_{i=1}^{3m} \Phi_i(\state_t) \le 150m^2 D^2+\sum_{i=1}^{3m}\mathcal{B}( \cload_0^{i} ).
\end{equation}
Let $T > 300m^4D^2 + \left\lceil16mD \ln(3km)\right\rceil$. Let us assume that $\state_{T}$ is not stable. Thus, for all $t \in \closedrange{t_0}{T}$,
$ \sum_{i=1}^{3m} \Phi_{i}(\state_t) - \sum_{i=1}^{3m} \Phi_{i}(\state_{t+2m^2}) \ge 1,$
and in particular:
$\sum_{i=1}^{3m} \Phi_{i}(\state_{t_0}) - \sum_{i=1}^{3m} \Phi_{i}(\state_{T}) \ge \left\lceil\frac{(T-t_0)}{2m^2} \right\rceil > 150m^2D^2,$
which contradicts with \eqref{eq:sum_of_pot}.

\end{proof}

\subsection{Lower bound}

We now give a lower bound on the stabilization time of parallel rotor-router walks.

\begin{theorem}
\label{lower_bound}
For any $N, M > 0$, $N \leq M \leq N^2$, there exists an initialization of the rotor router system in some graph with $\Theta(N)$ nodes and $\Theta(M)$ edges such that the stabilization time is $\Omega(M^2 \log N)$.
\end{theorem}
\begin{proof}
We start the proof by exhibiting for all $n>0$ an initial configuration with $k=2$ tokens on the $n$-node path, for which the stabilization time is $\Omega(n^2\log n)$. This construction then extends in a straightforward manner to denser graphs.

We will number the nodes of the path with consecutive integers $1,2,\ldots, n$, starting from its left endpoint. Initially, we assume that both of the tokens are located at node $\lceil n/3 \rceil$. Initially, the pointers of all nodes in the range $[2, \lceil n/3 \rceil]$ are directed towards the left endpoint of the path, and the pointers of all nodes in the range $[\lceil n/3 \rceil+1, n-1]$ are directed towards the right endpoint; the pointers of the degree-one nodes $1$ and $n$ are fixed in their unique position. We assign the tokens with unique identifiers, $1$ and $2$, defined so that the token labeled $1$ is always located not further from the left endpoint of the path than token $2$. The identity of a token is persistent over time, and the tokens may only swap identifiers when located at the same node of the path or when simultaneously traversing the same edge of the path in opposite directions.

For the considered initial configuration, the movement of the tokens in the first time steps is such that token $1$ is propagated towards the left endpoint of the path (reaching it after $\lceil n/3 \rceil -1$ steps), while token $2$ is propagated towards the right endpoint (reaching it after $t_0 = n-\lceil n/3 \rceil$ steps). After reaching their respective endpoints, the tokens bounce back and begin to move towards the opposite endpoint of the path. All nodes of the path have now been visited by some token. Following the approach introduced in~\cite{DBLP:conf/podc/KlasingKPS13}, for all subsequent moments of time $t$, we consider the partition $V = V_1(t) \cup V_2(t)$ of the set of nodes of the path into so called \emph{domains} $V_1(t)$ and $V_2(t)$, where $V_i(t)$ is the set of nodes of the path whose last visit up to time $t$ inclusive was performed by token $i$, $i=1,2$ (ties on the domain boundary are broken in favor of token $1$). As observed in~\cite{DBLP:conf/podc/KlasingKPS13}, each of the domains $V_i(t)$ is a sub-path of the considered path, containing token $i$ which traverses $V_i(t)$ between its two endpoints, enlarging its domain by one node of the neighboring domain each time it reaches the boundary of the two domains. In our case of $k=2$ tokens, we define the boundary point $b(t)$ so that $V_1(t) = [1,b(t)]$ and $V_2(t) = [b(t)+1, n]$. Initially, at the time $t_0$ when token $2$ reaches endpoint $n$ for the first time, we have $b(t_0) = \lceil n/3 \rceil$. We will consider the interval of time $[t_0, t_1]$, where $t_1$ is the first moment of time such that $b(t_1) = \lfloor n/2 \rfloor - 1$. Throughout the interval $[t_0, t_1]$, we have $|V_2(t)| > |V_1(t)|$. Thus, for any node $v$ such that $v = b(t)$ for some $t\in [t_0, t_1]$, if token $2$ visits $v$ at some time when heading towards the right endpoint of the path, node $v$ will be subsequently visited by token $1$ before token $2$ returns to $v$. It follows that the boundary of the domains can never move to the left by more than one step; formally, for any $t, t' \in [t_0, t_1]$, $t<t'$, we have $b(t') \geq b(t)-1$.

Now, let $\tau(v)$, for a node $v\in [\lceil n/3\rceil,\lfloor n/2 \rfloor]$, denote the first moment of time $\tau\in [t_0, t_1]$ such that the domain boundary has reached node $v$, \ie, $b(\tau) \geq v$. We will show a lower bound on the value of $\tau(\lfloor n/2\rfloor)$. We begin by lower-bounding the value $\Delta \tau_v = \tau(v+3) - \tau(v)$, for $v\in [\lceil n/3\rceil,\lfloor n/2 \rfloor  - 3]$. Clearly, $\tau(v+3) \geq \tau(v)$. During the considered time interval $(\tau(v),\tau(v+3)]$, the current domain boundary must be visited at least $3$ times more
by token $1$ than by token $2$. Since for any $t\in(\tau(v),\tau(v+3)]$, we have $b(t) \geq v-1$, and the distance between node $v-1$ and the left endpoint of the path is $(v-2)$, the total number of visits of token $1$ to the boundary during the considered time interval of length $\Delta \tau_v$ is at most $\lceil \Delta \tau_v / (2(v-2))\rceil$. On the other hand, since we also have $b(t) \leq v+3$ during the considered time interval, and the distance between node $v+3$ and the right endpoint of the path is $n-v-3$, we have that the domain boundary is visited by token $2$ at least $\lfloor \Delta \tau_v / (2(n-v-3))\rfloor$ times during the considered time interval. It follows that:
\begin{equation}
\left\lceil \frac{\Delta \tau_v}{2(v-2)}\right\rceil - \left\lfloor \frac{\Delta \tau_v}{2(n-v-3)}\right\rfloor \geq 3,\label{eq:lbound}
\end{equation}
which directly implies:
\begin{align*}
&\Delta \tau_v \geq \left(\frac{1}{2(v-2)} - \frac{1}{2(n-v-3)}\right)^{-1} > 2 \frac{(v-2)(n-v-3)}{n-2v} >\\&> 2 \frac{(n/3-2)(n/2-3)}{n-2v} = \frac{(n-6)^2}{6}\cdot\frac{1}{n/2 - v}.
\end{align*}
By summing the time increments $\Delta \tau_v$ over the subpath leading from $\lceil n/3\rceil$ to $\lfloor n/2 \rfloor$, with a ``step'' of three nodes, we obtain a lower bound on $\tau(\lfloor n/2\rfloor)$:
\begin{align*}
&\tau(\lfloor n/2\rfloor) \geq \sum_{a = 1}^{n/18 - 1} \Delta\tau_{\lfloor n/2 \rfloor  - 3a} \geq \frac{(n-6)^2}{6} \cdot \sum_{a = 1}^{n/18 - 1} \frac{1}{3a+1} >\\&> \frac{(n-6)^2}{18} \sum_{a = 2}^{n/18} a^{-1} > \frac{(n-6)^2}{18}(\ln n - 5).
\end{align*}
Thus, $\tau(\lfloor n/2\rfloor) = \Omega (n^2 \log n)$. It remains to note that $\tau(\lfloor n/2\rfloor)$ is a lower bound on the stabilization time of the rotor-router. Indeed, for as long as $t < \tau(\lfloor n/2\rfloor )$, the domain boundary satisfies $b(t) < \lfloor n/2\rfloor$, and the left endpoint of the path is visited by token $1$ every not more than $2(\lfloor n/2\rfloor -2) < n-1$ steps. On the other hand, in the limit cycle, each node (except endpoints) is visited by a token precisely twice in total during $2m = 2(n-1)$ steps, a contradiction. Hence, we obtain a lower bound of $\Omega (n^2 \log n)$ on the stabilization time for a configuration of two tokens on the path.

We now extend our construction to obtain graphs with higher edge density admitting a lower bound of $\Omega(m^2 \log n)$ on stabilization time. For given $N$ and $M$, we build such a graph $G$ starting with an $N$-node path $P_N$, with pointers along internal nodes of the path initialized as before. We then attach the path to a pair of identical cliques on $M'=\Theta(M)$ edges by identifying each endpoint of the path with one vertex of the respective clique. The initialization of the pointers within the clique is such that a token entering the clique from the path performs a Eulerian walk of exactly $2M'$ arcs on the clique, returning to the path after traversing each edge of the clique exactly twice and returning all pointers within the clique to their initial position. (Such an initialization of pointers on the clique always exists and corresponds to the stable state pointer arrangement of the single-token rotor-router, for the token located at the distinguished vertex.) In our considerations, the initial location of the tokens is assumed to be on node $\lceil N/3\rceil$ of the path, as before. The structure of domains within the path is defined analogously as before, and we can disregard the actions of tokens 1 and 2 within their respective cliques completely, replacing them for purposes of analysis by a delay of $2M'$ time steps during which the token is frozen at the left or right endpoint of the path, respectively. Adopting the same notation $\tau(v)$, $v\in [\lceil N/3\rceil,\lfloor N/2 \rfloor]$, as for the case of the path, we obtain analogously to~\eqref{eq:lbound}:

$$
\left\lceil \frac{\Delta \tau_v}{2M' + 2(v-2)}\right\rceil - \left\lfloor \frac{\Delta \tau_v}{2M' + 2(N-v-3)}\right\rfloor \geq 3,
$$
and from there (for $N\geq 9$):
$$
\Delta \tau_v \geq \frac{M'^2}{2} \frac{1}{N/2 - v}
$$
and finally:
$$
\tau(\lfloor N/2\rfloor) \geq \frac{M'^2}{6}(\ln N - 5).
$$
Since $\tau(\lfloor N/2\rfloor)$ is once again a lower bound on the stabilization time, and since our graph has $\Theta (N)$ nodes and $\Theta(M) = \Theta(M')$ edges, we have obtained the sought example of a rotor router configuration with a stabilization time of $\Omega(M^2 \log N)$.

\end{proof}

\section{Simulation of the Rotor-Router}

\label{sec:simulation}
In this section, we answer the question of how to efficiently query for the state of a parallel rotor-router system after a given number of steps. The result below is for an arbitrary number of tokens ($k \geq 1$).
For a single token ($k=1$) rotor-router mechanism we provide a faster simulation algorithm in the Appendix.

\begin{theorem}
\label{th:krr_query1}
We can preprocess any $\state_0$, in polynomial time and space (with respect to $n,m, \log k$) so that we can answer queries of  \emph{state} $\state_\tau$ or queries of $\cload_0^{\tau}(e)$ (the total number of visits until time step $\tau$) both in time $\bigo(n+m)$.
\end{theorem}

\begin{proof}
Our first step is to find $T$ such that $\state_T$ is stable. By Theorem~\ref{th:stabilizationtime} it is enough to take any $T > 300m^4D^2 + \left\lceil16mD \ln(3km)\right\rceil$.
We compute and keep states $\state_0,\state_1,\ldots,\state_T$, thus answering any queries of $\state_\tau$ with $\tau < T$ in $\bigo(n+m)$ time.
We store preprocessed $\cload_0^{\tau}(e)$ for any $\tau \in \closedrange{0}{T}$.

By Corollary~\ref{ob:subc}, we can find any valid $\infty$-subcycle decomposition of $\state_T$ in polynomial time.
By the properties of the subcycle decomposition, we can then find the values of $\state_{T+\tau}(e)$ by finding $e'$ being shifted by $\tau$ along the cycle $e$ belongs to. In a similar fashion we find $\cload_{T}^{T+\tau}(e)$ for each arc $e$, giving us $\cload_{T}^{T+\tau}(v)$ for each vertex $v$, thus we know the new pointer location for $v$. Each cycle can be preprocessed with prefix sums such that queries of this type can be answered in $\bigo(1)$ time, thus giving $\bigo(n+m)$ time for full $\state_{T+\tau}$ query.

We can preprocess each cycle with prefix sums, thus giving us the access to $\cload_{T}^{\tau}(e)$ for $\tau \in \halfrange{T}{\infty}$. By adding the value of $\cload_{0}^{T}(e)$ we get  desired $\cload_0^{\tau}(e)$.

\end{proof}

\subsection{Simulation of the single token rotor-router}
\paragraph{Computational model.}
In this part of the paper, we assume the standard Word RAM model, where the space is measured in machine words.
We assume that all numbers appearing in the input (i.e., $n$, $m$, the current time, and so on) fit in one machine word, and all basic
arithmetical operations (including indirect addressing) on such words take $\bigo(1)$ time.
\\
\\
We denote the node where the token starts as $v_0$, the node the token is located \emph{after} $i$ steps as $v_{i}$ and the arc the token traverses during the $i$-th step as $e_i = (v_{i-1},v_{i})$

We partition the sequence of moves into \emph{phases}, where the walk of the $i$-th \emph{phase}, denoted $\walk_i,$ ends with a move finishing the $i$-th \emph{full rotation} of starting node (thus $\walk_i$ being an Eulerian cycle for a subgraph of $\vec{G}$ containing all outgoing arcs of $v_0$).
Then the walk of the token can be denoted as:
$$ \walk_{\infty} = \underbracket{e_1,\ldots,e_{a_1}}_{\walk_1},\underbracket{e_{a_1+1},\ldots,e_{a_2}}_{\walk_2},\underbracket{e_{a_2+1},\ldots,e_{a_3}}_{\walk_3},\ldots $$.

We recall a consequence of main result from \cite{YanovskiWB03}.
\begin{proposition}
\label{proposition:1}
The sequence of arcs in a phase $i-1$ is always a subsequence of sequence from phase $i$. Moreover, if some arc is not visited during phase $i$, the inclusion is strict.
\end{proposition}
If we denote by $p$ the index of first phase that covers all of the arcs of the graph (thus $\walk_p$ being an Eulerian cycle on all arcs), the Proposition~\ref{proposition:1} can be rephrased as
$ \walk_1 \sqsubset \walk_2 \sqsubset \ldots \sqsubset \walk_p = \walk_{p+1} = \walk_{p+2} = \ldots $.
For any $v$, let $\expl(v)$ be the smallest $i$ such that $v$ is visited in the $i$-th phase.
Analogously the smallest $i$ such that $e \in \walk_i$ is denoted as $\expl(e)$. We also define $\expl(v_0)$ as $0$.
We define \emph{the exploratory stage} as maximal consecutive subsequences of $\walk_i$ containing arcs not being part of $\walk_{i-1}$.
The result of Yanovski \etal\ \cite{YanovskiWB03}, summarized in the next proposition, characterizes the structure of \emph{exploratory stages}.

\begin{proposition}
\label{proposition:3}
The graph induced in each \emph{exploratory stage} is Eulerian. In every phase, each arc is traversed at most once, and furthermore in every phase starting with the $(\expl(v)+1)$-th one, every arc adjacent to $v$ is traversed exactly once. For every two adjacent nodes $v$ and $u$ we have
that $|\expl(v) - \expl(u)| \le 1$, and actually for every node $w$ there exists a sequence of adjacent nodes  $w_0,w_1,w_2,\ldots,w_k$, where $w_0 = v_0$ and $w_k = w$, such that $ \expl(w_{i+1})-\expl(w_i) \in \{0,1\}$.
\end{proposition}

To be able to efficiently predict the position of the token at arbitrary time $T$, we need to efficiently discover the full exploration path
$\walk = \walk_1 \walk_2 \ldots \walk_p$.
However, $\walk$ might be too large to store it directly (it is of size $\bigo(m\cdot D)$), therefore we have to store its \emph{compressed form}:
the sequence of arcs from $\walk_p$ (the Eulerian cycle covering $G$) together with values of $\expl(e)$ for every arc $e$.

\begin{algorithm}
\caption{Simulation}
\label{alg:simulation}
\begin{algorithmic}[1]
\STATE reset rotor-router pointers
\STATE $\expl(e) \gets -1$ for every edge $e$
\STATE $Q \gets [v_{0}]$ \COMMENT{queue with one element}
\STATE $i \gets 0$
\WHILE{ $Q \neq \emptyset$}
	\STATE $i \leftarrow i+1$
	\STATE $R \leftarrow \emptyset$ \COMMENT{empty queue}
	\FORALL{$v \in Q$}
		\STATE $w \leftarrow \text{port}(v)$\COMMENT{do a rotor-router walk starting from $v$ until it reaches an already explored arc}
		\WHILE{$\expl[ (v, w) ] = -1$}
			\STATE $\expl[ (v,w) ] = i$
			\STATE $\text{advance}(v)$ \COMMENT{advance the pointer of the node $v$}
			\STATE  append $w$ to $R$
			\STATE $v \gets w$ \COMMENT{proceed to the next node}
			\STATE $w \gets \text{port}(v)$
		\ENDWHILE
	\ENDFOR
	\STATE $Q \leftarrow R$
\ENDWHILE
\end{algorithmic}
\end{algorithm}

Consider Algorithm~\ref{alg:simulation}. Its correctness follows from the fact that in consecutive iterations of the main loop, we simulate traversals of \emph{exploratory stages} from corresponding $\walk_i \setminus \walk_{i-1}$. In addition to finding the $\expl$ value for every arc, it results in the rotor-router pointers being in the same state as after the token performing a full $\walk$ walk.
We can use this property to find the sequence of arcs from $\walk_{p}$ by simply simulating the rotor-router walk starting from $v_0$ by another $2\cdot m$ steps.

\begin{proposition}
After running Algorithm~\ref{alg:simulation}, we can ask queries about position of token after $T$ steps in time $\bigo(m)$.
\end{proposition}

To speed up the above method, we denote $\walk_p = e'_1,e'_2, \ldots, e'_{2m}$, and let $e'_{2m+1}=e'_1$.
Then if we look at the sequence
of all values $\expl(e'_1), \expl(e'_2), \ldots, \expl(e'_{2m})$, they have the following property.

\begin{proposition}
\label{prop:small changes}
For all $i=1,\ldots,2m$ we have $|\expl(e'_i)-\expl(e'_{i+1})|\leq 1$.
\end{proposition}

We want to find the arc traversed by the token in the $T$-th step. For this we first check if $T$ exceeds $|\walk|$, and if so retrieve
the $((T-|\walk|)\bmod 2m)$-th element of $\walk_p$. Otherwise we need to locate the appropriate $\walk_i$ and its element. This reduces to finding the predecessor
of $T$ in the set $\{0,|\walk_1|,|\walk_1|+|\walk_2|,\ldots\}$ and then retrieving the $(T-\sum_{j<i}|\walk_j|)$-th element of $\walk_i$. We will show how to perform both operations
in $\bigo(\log\log m)$ time after linear preprocessing.

\begin{lemma}[$y$-fast trees~\cite{Willard}]
\label{lemma:y fast}
Any set $S\subseteq [1,U]$ can be preprocessed in $\bigo(|S|)$ space so that we can locate the predecessor of any $x$ in $S$ using $\bigo(\log\log  U)$ time.
\end{lemma}

Using the above lemma, the predecessor of $T$ in $\{0,|\walk_1|,|\walk_1|+|\walk_2|,\ldots\}$ can be found in $\bigo(\log\log m)$ time after linear preprocessing, as $U=\sum_i|\walk_i| = \bigo(mn)$.
Therefore we focus on preprocessing all $\walk_i$ as to allow retrieving the $k$-th element of any $\walk_i$ in $\bigo(\log\log m)$ time. This is not completely trivial,
as we want to keep the preprocessing space linear in $|\walk_p|$, not $\sum_i |\walk_i|$.

\begin{lemma}
\label{lemma:query}
All $\walk_i$ can be preprocessed in $\bigo(m)$ space so that we can retrieve the $k$-th element of any $\walk_i$ in $\bigo(\log\log m)$ time.
\end{lemma}

\begin{proof}
We decompose every $\walk_i$ into contiguous parts. Every part is a maximal fragment consisting of arcs which are adjacent in the final $\walk_p$. In other word, every arc such that its
predecessor on $\walk_p$ is different than on $\walk_i$ begins a new part. Because of Proposition~\ref{prop:small changes}, the number of such arcs on $\walk_i$ is bounded by $|\walk_{i+1}|-|\walk_i|$.
This is because if we have two arcs $e,e'$ which are neighbors on $\walk_i$ but not $\walk_p$, then there must be at least one arc between them on $\walk_{i+1}$. We create a predecessor
structure (implemented using Lemma~\ref{lemma:y fast}) storing the positions on $\walk_i$ of all such arcs starting a new fragment.
Additionally, the structure keeps for every such arc a pointer to the place where the fragment it begins occurs on $\walk_p$.
Then to retrieve the $k$-th element of $\walk_i$, we use the predecessor structure to find the part the answer belongs to. Then we return the corresponding element of $\walk_p$, which
takes $\bigo(1)$ time if we store $\walk_p$ in an array. Hence the total query time is $\bigo(\log\log m)$. To bound the total space, observe that it is equal to $\bigo(|\walk_p|+\sum_i |\walk_{i+1}|-|\walk_i|)=\bigo(|\walk_p|)=\bigo(m)$.

\end{proof}

Combining Algorithm~\ref{alg:simulation} and Lemma~\ref{lemma:query}, we get following result.
\begin{theorem}
\label{th:queries1}
Graph $G$ can be preprocessed in time $\bigo(n+m)$ such that we can find in time $\bigo(\log\log m)$ the position of the token after given time $T$.
\end{theorem}

Another natural problem that we would like to solve is how to preprocess $\walk$ using small space so that we can efficiently answer queries of the form: \emph{How many times did the token traversed the node $v$ in the first $T$ steps of exploration?}

We will use similar properties as the ones presented above.
\begin{proposition}
If $vis_i(v)$ is the number of times $v$ is visited during the $i$-th phase, then
$vis_i(v) = 0$ for all $i < \expl(v)$, $vis_i(v)\in (0,\deg(v)]$ for $i = \expl(v)$, and finally $vis_i(v) = \deg(v)$ for all $i>\expl(v)$.
\end{proposition}
\begin{theorem}
Graph $G$ can be preprocessed in time $\bigo(n+m)$ so that given a node $v$ and a number $T$ we can compute in time $\bigo(\log \log m)$ how many times
node $v$ was visited during the first $T$ steps.
\end{theorem}
\begin{proof}
We start with performing the preprocessing giving us the compressed representation of $\walk$ and $\expl$ value of every edge and every node. Then, for every node $v$ we will create a data structure allowing us to efficiently compute how many times $v$ was visited during the first $T$ steps.
To do this, for every node $v$ we will keep a triple
$(\expl(v), \mathcal{X}_v,  \mathcal{Y}_v)$,
where $\mathcal{X}_v$ is a predecessor structure with the timestamps of all visits to $v$ that happened in the $\walk_{\expl(v)}$-th phase, and $\mathcal{Y}_v$ is a predecessor structure with the timestamps of all visits to $v$ in the $\walk_p$-th phase.
To process the query, we first apply Theorem~\ref{th:queries1} to determine the position of the token after $T$ steps. Say that in the $T$-step the token traverses arc $e$ and is in the $I$-th exploratory phase. If $I < \expl(v)$, then we return $0$. If $I = \expl(v)$, then we compute and return the number of elements of $\mathcal{X}_v$ smaller than $T$, which can be done with a predecessor search.
Finally, if $I > \expl(v)$, we know that the total number of times $v$ was visited in all stages up to the $(I-1)$-th is $|\mathcal{X}_v| + (I-\expl(v)-1)\cdot \deg(v)$.
Hence we only need to count the visits to $v$ in the $I$-th phase. Because in every further phase the order in which the token visits $v$ and traverses $e$ is the same,
this can be done by counting how many times $v$ is visited in the $\walk_p$-th phase before traversing the arc $e$. After storing for every arc the time it is traversed
in the $\walk_p$-th phase, this can be done with a single predecessor search in $\mathcal{Y}_v$.

The size of predecessor structures kept for each node $v$ is $\bigo(\deg(v))$ and time necessary to create them is $\bigo(\deg(v))$, thus giving $\bigo(n+m)$ total preprocessing time and space.

\end{proof}

\section{Conclusion}

The rotor-router process has, in previous work, been identified as an efficient deterministic technique for a number of distributed graph processes, such as graph exploration and load balancing. In these settings, it rivals or outperforms the random walk, in some cases (such as parallel exploration of graphs) providing provable guarantees on performance, the counterparts of which need yet to be shown for the random walk. In this paper, we provide a complete characterization of the long-term behavior of the rotor-router, showing an inherent order in the limit state to which the system rapidly converges. This provides us with a better understanding of, e.g., the long-term load balancing properties of rotor-router-based algorithms, while at the same time opening the area for completely new applications. For instance, in view of our work, the rotor-router becomes a natural candidate for a self-organizing locally coordinated algorithm for the \emph{team patrolling problem} --- a task in which the goal is to periodically and regularly traverse all edges of the graph with $k$ agents. This topic, and related questions, such as bounding the maximum distance between tokens on their respective Eulerian cycles in the limit state of the rotor-router, are deserving of future attention.

\newpage
\pagenumbering{roman}
\bibliography{biblio}
\newpage
\appendix
\begin{center} \large{ \textbf{APPENDIX} } \end{center}

\section{Proofs of Lemmas and Theorems from Section~\ref{sec:lock-in}}



\noindent\textbf{Lemma~\ref{lem:pot_drop}.}
\emph{For arbitrary $i$ and $t$, the $i$-th potential is nonincreasing: $\Phi_{i}(\state_{t+1}) \le \Phi_{i}(\state_{t}).$}
\begin{proof}
Let us assume on the contrary, that there exist other minimal partition $A=\mset{a_1,\ldots,a_d}$.
As the partition into $\lfloor\frac{S}{d}\rfloor$ and $\lceil\frac{S}{d}\rceil$ is the only one with discrepancy at most 1, there exists pair of elements of $A$, that is w.l.o.g. $a_1$ and $a_2$ such that: $a_1 \ge a_2+2$.
Thus:
$$a_1^2 + a_2^2 > a_1^2 + a_2^2 - 2(a_1-a_2-1)  = (a_1-1)^2 + (a_2+1)^2,$$
giving that $\mset{ a_1-1, a_2+1, a_3, \ldots, a_d}$ has smaller sum of squares, which contradicts the assumption.
\end{proof}

\noindent\textbf{Lemma~\ref{lem:long}.}
\emph{The following statements are equivalent:
\begin{enumerate}[(i)]
\item $\state_{T}$ admits a $\dt$-step subcycle decomposition,
\item $\forall_v \forall_{  t,t' : \halfrange{t}{t'} \subseteq \halfrange{T}{T+\dt} } \mset{\cload_{t}^{t'}(e)}_{e\in \inedg(v)} = \mset{\cload_{t+1}^{t'+1}(e)}_{e\in \outedg(v)}$ (multisets of cumulative loads are preserved locally),
\item $\forall_{ t,t' :  \halfrange{t}{t'} \subseteq \halfrange{T}{T+\dt} } \mset{ \cload_{t}^{t'}(e)}_{e \in \vec{E}}  = \mset{ \cload_{t+1}^{t'+1}(e)}_{e \in \vec{E}}$ (multisets of cumulative loads are preserved globally),
\item $\forall_{0 \le i \le \dt} \Phi_i(\state_{T}) = \Phi_i(\state_{T+1}) = \ldots = \Phi_i(\state_{T+\dt-i+1})$ (potential is constant),
\item $\forall_v \forall_{e_1,e_2 \in \inedg(v)} \forall_{ t,t' :  \halfrange{t}{t'} \subseteq \halfrange{T}{T+\dt} }  |\cload_{t}^{t'}(e_1) - \cload_{t}^{t'}(e_2)| \le 1$ (incoming discrepancy is at most one).
\end{enumerate}}

\begin{proof}
(\ref{lem:long:item1}) $\Rightarrow$ (\ref{lem:long:item2})
A straightforward consequence of the $\dt$-step subcycle decomposition is existence of local bijective mapping $M_v : \inedg(v) \to \outedg(v)$ that:
$$\forall_{  t,t' : \halfrange{t}{t'} \subseteq \halfrange{T}{T+\dt} } \forall_{e \in \inedg(v)} \cload_{t}^{t'}(e) = \cload_{t+1}^{t'+1}(M_v(e)).$$

(\ref{lem:long:item2}) $\Rightarrow$ (\ref{lem:long:item3})
Taking union over all vertices, we get:
\begin{align*}\forall_{  t,t' : \halfrange{t}{t'} \subseteq \halfrange{T}{T+\dt} } \mset{\cload_{t}^{t'}(e)}_{e\in \vec{E}} =& \bigcup_{v \in V} \mset{\cload_{t}^{t'}(e)}_{e\in \inedg(v)} = \bigcup_{v \in V} \mset{\cload_{t+1}^{t'+1}(e)}_{e\in \outedg(v)} =\\ =& \mset{\cload_{t+1}^{t'+1}(e)}_{e\in \vec{E}},
\end{align*}
from which the claim follows.

(\ref{lem:long:item3}) $\Rightarrow$ (\ref{lem:long:item4})
Claim follows in a straightforward fashion from the definition of the $\Phi_i$.

(\ref{lem:long:item4}) $\Rightarrow$ (\ref{lem:long:item5})
Assume that, to the contrary, there exists $v'$ and $e_1,e_2 \in \inedg(v')$ such that:
$$ \cload_t^{t'}(e_1) - \cload_t^{t'}(e_2) \ge 2.$$
Let $i = t'-t$ and let us define:
\begin{align*}
c(e_1) &= \cload_t^{t'}(e_1)-1,\\
c(e_2) &= \cload_t^{t'}(e_2)+1,\\
c(e) &= \cload_t^{t'}(e) \text{ for every other } e.
\end{align*}
We observe, that:
$$\forall_v \sum_{e \in \inedg(v)} \cload_t^{t'}(e) = \sum_{e \in \inedg(v)} c(e).$$

Thus by the same reasoning as in proof of Lemma~\ref{lem:pot_drop}:
$$\Phi_i(\state_t) = \sum_v \sum_{e \in \outedg(v)}  \cload_t^{t'}(e)^2 > \sum_v \sum_{e \in \outedg(v)} c(e)^2 = \sum_v \sum_{e \in \inedg(v)} c(e)^2 \ge \Phi_i(\state_{t+1}),$$
which contradicts the assumption.

(\ref{lem:long:item5}) $\Rightarrow$ (\ref{lem:long:item1})
Let us fix $v$ and denote $\inedg(v) = \{e_1,e_2,\ldots,e_d\}$, $\outedg(v) = \{f_1,f_2,\ldots,f_d\}$.
Under the assumption, we deduce the existence of permutation $(\pi_1,\ldots,\pi_d)$ such that:
$$
\forall_{0 \le i \le \dt} \cload_T^{T+i}(e_{(\pi_1)}) \ge \cload_T^{T+i}(e_{(\pi_2)}) \ge \ldots \ge \cload_T^{T+i}(e_{(\pi_d)}),
$$
because otherwise, we would get that there exists $e_1,e_2,i_1<i_2\le\dt$, such that:
$$\cload_T^{T+i_1}(e_1) < \cload_T^{T+i_1}(e_2),\cload_T^{T+i_2}(e_1) > \cload_T^{T+i_2}(e_2).$$
implying that
$$\cload_{T+i_1}^{T+i_2}(e_1) - \cload_{T+i_1}^{T+i_2}(e_2) \ge 2,$$
a contradiction with our initial assumption.

There exists a permutation $(\sigma_1),(\sigma_2),\ldots,(\sigma_d)$ such that $f_{(\sigma_1)},f_{(\sigma_2)},\ldots,f_{(\sigma_d)}$ have the same cyclic ordering as $\rho_v$, and $f_{(\sigma_1)}$ is the arc pointed to by $\pointer_v$ at beginning of time step $T+1$, along which the first token at time $T+1$ is propagated. From the cyclic nature of the rotor-router mechanism, we obtain:
$$
\forall_{0 \le i \le \dt} \cload_{T+1}^{T+i+1}(f_{(\sigma_1)}) \ge \cload_{T+1}^{T+i+1}(f_{(\sigma_2)}) \ge \ldots \ge \cload_{T+1}^{T+i+1}(f_{(\sigma_d)}).
$$

Thus, the sequences $\left(\cload_{T}^{T+i}(e_{(\pi_1)}),\ldots,\cload_{T}^{T+i}(e_{(\pi_d)})\right)$ and  \\$\left(\cload_{T+1}^{T+i+1}(f_{(\sigma_1)}),\ldots,\cload_{T+1}^{T+i+1}(f_{(\sigma_d)})\right)$ are identical since they are sorted sequences of integers, with identical sum and discrepancy not greater than 1 (thus being composed of only two different values).

But $\forall_{0 \le i \le \dt} \cload_{T}^{T+i}(e_{(\pi_j)}) = \cload_{T+1}^{T+i+1}(f_{(\sigma_j)})$ implies $\load_{T}(e_{(\pi_j)}) = \load_{T+1}(f_{(\sigma_j)})$, $\load_{T+1}(e_{(\pi_j)}) = \load_{T+2}(f_{(\sigma_j)})$, $\ldots$, $\load_{T+\dt-1}(e_{(\pi_j)}) = \load_{T+\dt}(f_{(\sigma_j)})$,
meaning that if we pair each of $e_{(\pi_1)},e_{(\pi_2)},\ldots$ with corresponding $f_{(\sigma_1)},f_{(\sigma_2)},\ldots$ we will obtain a subcycle decomposition.

\end{proof}

\noindent\textbf{Theorem~\ref{th:subcycle}.}
\emph{The following conditions are equivalent:
\begin{enumerate}[(i)]
\item $\state_T$ is stable,
\item $\state_T$ admits a $\infty$-subcycle decomposition,
\item $\state_T$ admits a $(2m^2)$-subcycle decomposition.
\end{enumerate}
}
\begin{proof}
$(\ref{subcycle2}) \Rightarrow (\ref{subcycle1})$
If all tokens traversing arcs follow Eulerian cycles as their routes, then we denote $t_r = \text{lcm}(|\vec{E_1}|,|\vec{E_2}|,\cdots,|\vec{E_p}|)$.

Since $\forall_{e \in \vec{E}}\load_T(e) = \load_{T+t_r}(e)$ by assumption, we easily observe that $\forall_{v \in V}\load_T(v) = \load_{T+t_r}(v)$.

To prove that the states are equivalent, we need to show that the rotor-router pointers for each vertex are in the same position at time steps $T$ and $T+t_r$. This is equivalent to:
$\forall_{e_1,e_2 \in \outedg(v)} \cload_T^{T+t_r}(e_1) = \cload_T^{T+t_r}(e_2).$

However, if we assume otherwise, then there exist two arcs $e_1,e_2$ sharing starting points, such that: $|\cload_T^{T+t_r}(e_1) - \cload_T^{T+t_r}(e_2)| \ge 1$. But that would imply that by taking the cyclic sum twice, $|\cload_T^{T+2\cdot t_r}(e_1) - \cload_T^{T+2\cdot t_r}(e_2)| \ge 2$ which contradicts the rotor-router property.

$(\ref{subcycle1}) \Rightarrow (\ref{subcycle3})$ Let us take $T'$ such that $\state_{T'} = \state_T$ and $T' \ge T+2m^2$.
From Lemma~\ref{lem:pot_drop} we deduce:
$
\forall_{i\ge 0} \Phi_{i}(\state_{T}) = \Phi_{i}(\state_{T+1}) = \Phi_{i}(\state_{T+2}) = \ldots = \Phi_{i}(\state_{T'}).
$
By using equivalence of clauses \eqref{lem:long:item4} and \eqref{lem:long:item1} in Lemma~\ref{lem:long} , we get the $2m^2$-subcycle decomposition for $\state_T$.

$(\ref{subcycle3}) \Rightarrow (\ref{subcycle2})$ Assume that $\state_T$ which has a $(2m^2)$-subcycle decomposition does not admit a $\infty$-subcycle decomposition. We denote by $\tau$ the largest integer such that $\state_T$ admits a $\tau$-subcycle decomposition ($\tau \ge 2m^2$).
By the maximality of $\tau$, taking into account the equivalence of clauses \eqref{lem:long:item1} and \eqref{lem:long:item5} in Lemma~\ref{lem:long}, there exists a vertex $v$ and two incoming arcs $e_1,e_2 \in \inedg(v)$ such that:
\begin{equation}
\label{eq:bad_discrp}
\exists_{T' \in \closedrange{T}{T+\tau}}  |\cload_{T'}^{T+\tau+1}(e_1) - \cload_{T'}^{T+\tau+1}(e_2)| \ge 2.
\end{equation}

Let $\vec{E}_1,\vec{E}_2$ be the Eulerian cycles from the $\tau$-subcycle decomposition (starting at state $\state_T$), such that $e_1 \in \vec{E}_1, e_2 \in \vec{E}_2$.

Let us denote $\ell = lcm(|\vec{E}_1|,|\vec{E}_2|)$.
Observe that $\ell \le m^2$, since we either have $\vec{E}_1 \cap \vec{E}_2 = \emptyset$, and then $|\vec{E}_1|+|\vec{E}_2| \le 2m$ implies that $\ell \le |\vec{E}_1|\cdot|\vec{E}_2| \le m^2$, or $\vec{E}_1 = \vec{E}_2$, and then $\ell = |\vec{E}_1| \le 2m \le m^2$.\footnote{$2m \le m^2$ requires $m\ge 2$, however in graphs with $m=1$ \emph{every} state is stable.}

Let us observe, that since the tokens are traversing the arcs over the Eulerian subcycles, we have, for any $\halfrange{t}{t+\ell} \subseteq \closedrange{T}{T+\tau}$:
\begin{equation}
\label{eq:cycle1}
\cload_{t}^{t+\ell}(e_1) = \frac{\ell}{|\vec{E}_1|} \sum_{e \in \vec{E}_1} \load_T(e),
\end{equation}
\begin{equation}
\label{eq:cycle2}
\cload_{t}^{t+\ell}(e_2) = \frac{\ell}{|\vec{E}_2|} \sum_{e \in \vec{E}_2} \load_T(e).
\end{equation}

Now we proceed to show that the left-hand sides of the above expressions have to be equal.
Indeed, observe that $2\ell\leq 2m^2$ and that condition \eqref{lem:long:item5} from Lemma~\ref{lem:long} also holds for range step $\halfrange{T}{T+2\ell}$, thus we have:
$|\cload_T^{T+2\ell}(e_1) - \cload_T^{T+2\ell}(e_2)| \leq 1,$
and moreover, we can write:
$$|\cload_T^{T+2\ell}(e_1) - \cload_T^{T+2\ell}(e_2)| = \left|\frac{2\ell}{|\vec{E}_1|}\cdot\sum_{e \in \vec{E}_1} \load_t(e) - \frac{2\ell}{|\vec{E}_2|}\cdot\sum_{e \in \vec{E}_2} \load_t(e)\right| = 2 \cdot |\cload_T^{T+\ell}(e_1) - \cload_T^{T+\ell}(e_2)|.$$
It follows that $\cload_T^{T+\ell}(e_1) = \cload_T^{T+\ell}(e_2)$, and taking into account \eqref{eq:cycle1} and \eqref{eq:cycle2} we get that for any $\halfrange{t}{t+\ell} \subseteq \halfrange{T}{T+\tau+1}$:
\begin{equation}\label{eq:shifting}\cload_{t}^{t+\ell}(e_1) = \cload_{t}^{t+\ell}(e_2).\end{equation}

Thus, the average number of tokens observed over any multiplicity of the cycle length has to be identical for both arcs $e_1,e_2$. We can use this fact to put a bound on discrepancy of number of tokens between both arcs, by analyzing last $\ell$ steps of the time segment separately.
Indeed, by invoking \eqref{eq:shifting} we get (with $T'$ chosen as in \eqref{eq:bad_discrp}),
if $T' \le T+\tau+1-\ell$ then:
$ |\cload_{T'}^{T+\tau+1}(e_1) - \cload_{T'}^{T+\tau+1}(e_2)|
 \stackrel{\eqref{eq:shifting}}{=} |\cload_{T'}^{T+\tau+1-\ell}(e_1)   - \cload_{T'}^{T+\tau+1-\ell}(e_2) |\le 1,$
and otherwise
$ |\cload_{T'}^{T+\tau+1}(e_1) - \cload_{T'}^{T+\tau+1}(e_2)|
 \stackrel{\eqref{eq:shifting}}{=} |\cload_{T'-\ell}^{T+\tau+1-\ell}(e_1)   - \cload_{T'-\ell}^{T+\tau+1-\ell}(e_2) |\le 1,$
and we see that $\tau$-subcycle decomposition contradicts with \eqref{eq:bad_discrp}.

\end{proof}

\noindent\textbf{Theorem~\ref{th:stable_charact}.}
\emph{State $\state_T$ is stable iff:
$$\sum_{i=1}^{3m} \Phi_{i}(\state_T) = \sum_{i=1}^{3m} \Phi_{i}(\state_{T+2m^2}).$$}

\begin{proof}
$\Rightarrow \text{ (only if) }$
This is a straightforward result of characterizations from Theorem~\ref{th:subcycle} and preservation of potentials (Lemma~\ref{lem:long}, (\ref{lem:long:item1}) $\Rightarrow$ (\ref{lem:long:item4})).

$\Leftarrow \text{ (if) }$
By Theorem~\ref{th:subcycle} it is enough to prove that $\state_T$ admits a $2m^2$-subcycle decomposition. \footnote{Once again we can eliminate case of $m=1$ and assume that $2m^2 \ge 3m$.}
Let us assume otherwise, that $\state_T$ doesn't admit $2m^2$-subcycle decomposition, and let $\tau$ be the largest integer such that there $\state_T$ admits $\tau$-subcycle decomposition ($\tau < 2m^2$).
Similarly to the proof of Theorem~\ref{th:subcycle}, there exists a vertex $v$ and two incoming arcs $e_1,e_2 \in \inedg(v)$, such that:
\begin{equation}
\label{eq:bad_discrp2}
\exists_{t\in \closedrange{T}{T+\tau}}  \left|\cload_{t}^{T+\tau+1}(e_1) - \cload_{t}^{T+\tau+1}(e_2)\right| \ge 2.
\end{equation}

Let us choose a local mapping $M$ valid with respect to the $\tau$-subcycle decomposition (starting at state $\state_T$), and let
  $\vec{E}_1,\vec{E}_2$ be the Eulerian cycles induced by $M$,  such that $e_1 \in \vec{E}_1, e_2 \in \vec{E}_2$.
We will look into the structure of sequences
\begin{align*}
(a_0,a_1,\ldots) &= (\load_T(e_1),\load_{T+1}(e_1),\ldots),\\
(b_0,b_1,\ldots) &= (\load_T(e_2),\load_{T+1}(e_2),\ldots),\\
(c_0,c_1,\ldots) &= (a_0-b_0,a_1-b_1,\ldots).
\end{align*}

Since each separate potential is nonincreasing, we deduce that $\forall_{0 \le i \le 3m} \Phi_i(\state_T) = \Phi_i(\state_{T+1}) = \ldots = \Phi_i(\state_{T+2m^2})$. Thus, all of $\state_T,\state_{T+1},\ldots,\state_{T+2m^2-3m}$ admit a $3m$-subcycle decomposition. From this we deduce, that:
\begin{equation}
\label{eq:sequence_property}
\forall_{\closedrange{i}{j} \subseteq \halfrange{0}{2m^2}} \text{ if } (j-i<3m) \text{ then } |c_i +c_{i+1}+\ldots+c_{j}| \le 1,
\end{equation}
in particular, $c_i \in \{-1,0,1\}$ for $0 \le i < 2m^2$.

From the $\tau$-subcycle decomposition we get that:
$$\forall_{\closedrange{i}{j} \subseteq \halfrange{0}{\tau}} |c_{i}+c_{i+1}+\ldots+c_{j}| \le 1$$

and rewriting \eqref{eq:bad_discrp2}:
\begin{equation}
\label{eq:bad_discrp3}
\exists_{\eta \in \closedrange{0}{\tau}} |c_\eta+ c_{\eta+1} + \ldots + c_{\tau}| \ge 2.
\end{equation}


Before we proceed further with the proof, we need following structural lemma on the words $a,b,c$:
\begin{lemma}
\label{lem:subwords}
Let $(a_i), (b_i)$ and $(c_i)$ be defined as in Theorem~\ref{th:stable_charact}.
If $c_i = 0$ holds for $2m$ consecutive values of $i$ from $0,\ldots,\tau$, then it holds for all of them.
\end{lemma}

\begin{proof}
$(a_i)_{i=0}^{i\le \tau}$ and $(b_i)_{i=0}^{i\le \tau}$ ($a$ and $b$ for short) are periodic sequences with short periods (2m at most).
Let $p_1=|\vec{E}_1|,p_2=|\vec{E}_2|$ be the periods lengths of those sequences.
We want to prove that if they are identical on corresponding fragments, they are equal everywhere.
We look into two cases:
\begin{itemize}
\item If $\vec{E}_1 = \vec{E}_2$, then $|p_1| = |p_2| \le 2m$. By $a$ and $b$ being equal on fragment of length $2m$, we get $p_1 = p_2$, giving $a=b$.
\item If $\vec{E}_1 \not= \vec{E}_2$, then we get $|p_1|+|p_2| \le 2m$. By application of Fine~and~Wilf’s theorem (if two periodic words with period lengths $x$ and $y$ are equal on fragment of length $x+y-\mathrm{gcd}(x,y)$, they are identical, see~\cite{FineWilf}), we get that $a=b$.
\end{itemize}

\end{proof}

By Lemma~\ref{lem:subwords} we get, that in any sequence that satisfies \eqref{eq:bad_discrp3} there cannot be a subsequence of $2m$ consecutive values of 0. However, from \eqref{eq:sequence_property}, we deduce that any two occurrences of $1$ (respectively $-1$) at distance less than $3m$ have to be separated by at least one occurrence of $-1$ (respectively $1$). Thus any occurrences of $1$ and $-1$ have to happen in an alternating fashion (separated by arbitrary number of zeroes), leading to contradiction with \eqref{eq:bad_discrp3}.

\end{proof}

\newpage
\section{Proofs of Lemmas and Theorems from Section~\ref{sec:stab}}
\label{appendix:sec4}
In a bipartite graph $G$ the two groups of tokens, the one traversing at one particular moment from one side to another and the one traversing in the opposite direction, will never met and exchange tokens. Thus it is possible to have arbitrarily large difference in size of groups traversing arcs at arbitrarily late timesteps. Thus we need to refine our definitions and theorems used in the load balancing analysis.

\paragraph{Balancing sets.} Now we introduce the partition of arcs into \emph{balancing sets}. Arcs belonging to the same set will balance the load over time, while no load balancing happens between different sets. We will formalize those notions in the following.

\noindent\textbf{Definition~\ref{def:distance}.}
\emph{We define the notion of \emph{distance} between two arcs $e,e'$, denoted by $d(e,e')$, as a length of shortest alternating path having $e$ and $e'$ as first and last arcs. If there is no such path, we say that $d(e,e') = \infty$.
We say two arcs $e,e'$  belong to the same \emph{balancing set} $\mathcal{P}$ iff  $d(e,e') < \infty$.}

\noindent\textbf{Lemma~\ref{lem:even_length}.}
\emph{The set of arcs of any bipartite graph is partitioned into exactly two balancing sets, while for every other graph the partition results in one balancing set. Moreover:
$$\forall_{e,e' \in \mathcal{P}} \ d(e,e') \le 4D+1$$
and for any such $e,e'$ there exists an alternating path connecting $e$ and $e'$ containing at most $2D$ pairs of arcs sharing ending vertices and at most $2D+1$ pairs of arcs sharing starting vertices.}

\begin{proof}
Since every arc in $G$ has opposite arc, any simple path in $G$ relates to an alternating path. Thus, observe that two arcs are in the same balancing set iff there exists path of even length (in terms of edges) connecting the respective starting vertices $v$ of arc $e$ and $v'$ of arc $e'$. Moreover, the length of such path increased by one gives an upper bound on $d(e,e')$.
Since in a bipartite graph any simple path connecting vertices from the same side is of even length, this instantly gives us bound $d(e,e') \le D+1$ for bipartite graphs when $e$ and $e'$ go in the same direction, and $d(e,e')=\infty$ when they go in opposite directions.

To solve the non-bipartite case, let $C$ be the shortest simple cycle of odd length in $G$ (it always exists in non-bipartite graphs). It is easy to show that $|C| \le 2D+1$. Let $u$ be the closest vertex on $C$ to $v$, and $u'$ to $v'$. Notice that $u$ and $u'$ can be connected by paths of arbitrary parity by using one of two parts of $C$ as a connection. Thus, there exists a path connecting $v$ to $v'$ going through vertices $v,u,u',v'$ (in this order) of even length at most $4D+1$. Thus the length of the path is at most $4D$. \begin{itemize} \item If the path doesn't contain both $e$ and $e'$, we get desired bound $d(e,e') \le 4D+1$ with number of pairs sharing starting vertices bounded by $2D+1$ and pairs sharing ending vertices bounded by $2D$. \item If the path contains at least one of $e$ or $e'$, then  $d(e,e') \le 4D$ giving also bounds on number of pairs of both types being at most $2D$.
\end{itemize}

\end{proof}

We will denote the value of the potential achieved by distribution of tokens that preserves sums of loads over load balancing sets of arcs (ignoring the restriction that loads are integers) by:
\begin{equation}\label{eq:b'}\mathcal{B}(\wload_t) \defeq \sum_{\mathcal{P}}|\mathcal{P}|\cdot\left( \avg_{e \in \mathcal{P}} \wload_t(e)\right)^{\!2},\end{equation}
where the sum is taken over all possible balancing sets $\mathcal{P}$.
Recall that for non-bipartite graphs \eqref{eq:b'} reduces to:
$$\mathcal{B}(\wload_t) = \frac{k^2}{2m} = \const,$$
which is consistent with our previous definition,
while for bipartite graphs (where $k_1$ and $k_2$ are number of tokens in both parts of graph)
$$\mathcal{B}(\wload_t) = \frac{k_1^2+k_2^2}{m} = \const.$$

\noindent\textbf{Definition~\ref{def:discr}.}
\emph{We say that a configuration of tokens $\wload_t$ has discrepancy over arcs equal to $$\max_{\mathcal{P}} \max_{e,e' \in \mathcal{P}} (\wload_t(e) - \wload_t(e'))= x.$$}

\noindent\textbf{Lemma~\ref{lem:potdrop}.}
\emph{Let at time step $t$ a distribution of tokens $\mathcal{W}_t$ has discrepancy over arcs equal to $x$, for $x>4D+1$. Then:
$$\Phi(\wload_t) - \Phi(\wload_{t+1}) \ge \frac{(x-4D-1)(x-1)}{4D}.$$}

\begin{proof}
From Lemma~\ref{lem:even_length} we know, that there exists an alternating path $e=e_0,e_1,e_2,\ldots ,e_{h-1},e_h=e'$ of length at most $4D+2$, such that each vertex belongs to at most 2 of those arcs, and that $e$ and $e'$ have discrepancy at least $x$. Thus:
$$\sum_{j = 0}^{h-1} |\wload_t(e_j)-\wload_t(e_{j+1})| \ge x.$$

Moreover, at most $2D$ of pairs $(e_j$,$e_{j+1})$ share ending vertices and at most $2D+1$ share starting vertices. Since by \eqref{eq:roundfair} loads of arcs sharing starting vertex can differ by at most 1, we get that there exist $I \le 2D$ and vertices $v_1,\ldots,v_{I}$, such that each vertex $v_i$ admits incoming load discrepancy $\delta_i$ (two incoming arcs $e_{i1}$ and $e_{i2}$ such that $\wload_t(e_{i1})-\wload_t(e_{i2}) \ge \delta_i$), and that:
$$\sum_{i\le I} \delta_i \ge x-2D-1.$$

Observe, that:
\begin{align*}
&\sum_{e \in \inedg(v_i)} (\wload_t(e))^2 = (\wload_t(e_{i1}))^2+(\wload_t(e_{i2}))^2+\sum_{e \in \inedg(v_i)\setminus\{e_{i1},e_{i2}\}} (\wload_t(e))^2 \ge\\
&\ge \left(\wload_t(e_{i1})-\left\lfloor\frac{\delta_i}{2}\right\rfloor\right)^2+\left(\wload_t(e_{i2})+\left\lfloor\frac{\delta_i}{2}\right\rfloor\right)^2 + \sum_{e \in \inedg(v_i)\setminus\{e_{i1},e_{i2}\}} (\wload_t(e))^2 \ge\\
&\ge \sum_{e \in \outedg(v_i)} (\wload_{t+1}(e))^2,
\end{align*}
which gives
$$\sum_{e \in \inedg(v_i)} (\wload_t(e))^2 - \sum_{e \in \outedg(v_i)} (\wload_{t+1}(e))^2 \ge (\wload_t(e_{i1}))^2+(\wload_t(e_{i2}))^2-\left(\wload_t(e_{i1})-\left\lfloor\frac{\delta_i}{2}\right\rfloor\right)^2-$$
$$-\left(\wload_t(e_{i2})+\left\lfloor\frac{\delta_i}{2}\right\rfloor\right)^2=
2\left(\wload_t(e_{i1})-\wload_t(e_{i2})\right)\left\lfloor\frac{\delta_i}{2}\right\rfloor-2\left\lfloor\frac{\delta_i}{2}\right\rfloor^2 \ge 2 \left\lceil\frac{\delta_i}{2}\right\rceil\left\lfloor\frac{\delta_i}{2}\right\rfloor \ge \frac{\delta_i^2-1}{2},$$
Taking a sum over all vertices, we get:
$$\Phi(\wload_t) - \Phi(\wload_{t+1}) \ge \sum_{i \le I} \frac{\delta_i^2-1}{2} \ge \left(\sum_{i \le I} \frac{\delta_i^2}{2}\right)  - \frac{I}{2} \ge \frac{(\sum_{i \le I} \delta_i)^2}{2I}-\frac{I}{2} \ge $$
$$\ge\frac{(x-2D-1)^2}{4D} - D=\frac{(x-4D-1)(x-1)}{4D}.$$

\end{proof}

\noindent\textbf{Lemma~\ref{lem:atleastdiscrp}.}
\emph{If $\wload_t$ has discrepancy $x$, then:
$\Phi(\wload_t) \le \mathcal{B}(\wload_0) + \frac12 m x^2.$}

\begin{proof}
Observe, that by maximizing the sum of squares while preserving the sum and the upper bound on discrepancy, we get that:
$$\sum_{e \in \mathcal{P}} \wload_t(e)^2 \le \frac{|\mathcal{P}|}{2} \cdot (\avg_{e \in \mathcal{P}}(\wload_t(e))+x/2)^2+\frac{|\mathcal{P}|}{2} \cdot (\avg_{e \in \mathcal{P}}(\wload_t(e)) - x/2)^2.$$
Summing above over all $\mathcal{P}$:
$$\Phi(\wload_t) \le \mathcal{B}(\wload_0) + \sum_{\mathcal{P}} |\mathcal{P}|\left(\frac{x}{2}\right)^2.$$

\end{proof}

\end{document}